\documentclass[graybox]{svmult}

\usepackage{type1cm}        
\usepackage{makeidx}         
\usepackage{graphicx}        
\usepackage{multicol}        
\usepackage[bottom]{footmisc}
\usepackage{newtxtext}       %
\usepackage{newtxmath}       
\usepackage{algorithmicx}
\usepackage{algpseudocode}
\usepackage{xifthen}
\usepackage{needspace}
\usepackage{sidecap}

\newtheorem{assumption}{Assumption}
\newcounter{assum}[assumption]
\newcommand{\Assumption}{\refstepcounter{assum}\theassumption~}

\newcommand{\mb}{\mathbf}
\newcommand{\mc}{\mathcal}
\newcommand{\mbb}{\mathbb}
\newcommand{\mr}{\mathrm}
\newcommand{\mf}{\mathfrak}
\newcommand{\Init}{\mr{Init}}
\newcommand{\KeyGen}{\mr{KeyGen}}
\newcommand{\Eval}{\mr{Eval}}
\newcommand{\Add}{\mr{Add}}
\newcommand{\Mult}{\mr{Mlt}}
\newcommand{\cMult}{\mr{cMlt}}
\renewcommand{\E}{\mr{E}}
\renewcommand{\D}{\mr{D}}
\newcommand{\hatKeyGen}{\mr{\hat{KeyGen}}}
\newcommand{\hatE}{\mr{\hat{E}}}
\newcommand{\hatD}{\mr{\hat{D}}}
\newcommand{\hatC}{\mr{\hat{\mc{C}}}}
\newcommand{\LabHE}{$\mr{LabHE}$ }
\newcommand{\AHE}{$\mr{AHE}$ }

\definecolor{aquamarine}{RGB}{67,167,133}

\newcommand{\numberthis}{\addtocounter{equation}{1}\tag{\theequation}}

\DeclareMathOperator*{\argmin}{arg\,min}

\newcounter{protocol}
\newenvironment{protocol}[1][]%
  {
    \needspace{2\baselineskip}
    \noindent \rule{\linewidth}{1pt} \endgraf
    \refstepcounter{protocol}
    \centering \textsc{Protocol}~\theprotocol%
    \ifthenelse{\isempty{#1}}{}{:\ #1}
  }{
  \vspace{-2pt}
  \noindent \rule{\linewidth}{1pt}\vspace{2pt}
  }

  \newcommand{\StatexIndent}[1][3]{%
  \setlength\@tempdima{\algorithmicindent}%
  \Statex\hskip\dimexpr#1\@tempdima\relax}

\makeindex             

\begin{document}

\title*{Secure Multi-party Computation for Cloud-based Control}
\author{Andreea B. Alexandru and George J. Pappas}
\institute{Andreea B. Alexandru,  George J. Pappas \at University of Pennsylvania, Department of Electrical Engineering, Philadelphia, PA, USA, \email{aandreea@seas.upenn.edu, pappasg@seas.upenn.edu}}
\maketitle

\abstract{In this chapter, we will explore the cloud-outsourced privacy-preserving computation of a controller on encrypted measurements from a (possibly distributed) system, taking into account the challenges introduced by the dynamical nature of the data. The privacy notion used in this work is that of cryptographic multi-party privacy, i.e., the computation of a functionality should not reveal \textit{anything} more than what can be inferred only from the inputs and outputs of the functionality. The main theoretical concept used towards this goal is Homomorphic Encryption, which allows the evaluation of sums and products on encrypted data, and, when combined with other cryptographic techniques, such as Secret Sharing, results in a powerful tool for solving a wide range of secure multi-party problems. We will rigorously define these concepts and discuss how multi-party privacy can be enforced in the implementation of a Model Predictive Controller, which encompasses computing stabilizing control actions by solving an optimization problem on encrypted data.}

\section{Introduction}
\label{sec:introduction}

Cloud computing has become a ubiquitous tool in the age of big data and geographically-spread systems, due to the capabilities of resource pooling, broad network access, rapid elasticity, measured service and on-demand self-service, as defined by NIST~\cite{Mell2011nist}. 
The computational power and storage space of a cloud service can be distributed over multiple servers. 
Cloud computing has been employed for machine learning applications in e.g., healthcare monitoring and social networks, smart grid control and other control engineering applications, and integration with the Internet of Things paradigm~\cite{Rittinghouse2016cloud,Botta2016integration}. However, these capabilities do not come without risks. The security and privacy concerns of cloud computing range from communication security to leaking and tampering with the stored data or interfering with the computation, that can be maliciously or unintentionally exploited by the cloud provider or the other tenants of the service~\cite{Ali2015security,Singh2016survey,Hamlin2016cryptography}. In this chapter, we will focus on concerns related to the privacy of the data and computation. These issues can be addressed by the cryptographic tools described below.

\textbf{Secure Multi-Party Computation} (SMPC) encompasses a range of cryptographic techniques that facilitate joint computation over secret data distributed between multiple parties, that can be both clients and servers. The goal of SMPC is that each party is only allowed to learn its own result of the computation, and no intermediary results such as inputs or outputs of other parties or other partial information. The concept of SMPC originates from~\cite{Yao82}, where a secure solution to the millionaire's problem was proposed. Surveys on SMPC can be found in~\cite{Cramer15}. 
SMPC involves communication between parties and can include individual or hybrid approaches between techniques such as secret sharing~\cite{Shamir1979share,Pedersen91,Beimel2011secret}, oblivious transfer~\cite{Naor2001efficient,Rabin2005exchange}, garbled circuits~\cite{Yao82,Goldreich87, Bellare12}, (threshold) homomorphic encryption~\cite{Rivest78data,Gentry2009fully,Naehrig2011can} etc. 

\textbf{Homomorphic Encryption} (HE), introduced in~\cite{Rivest78data} as \textit{privacy homomorphism}, refers to a secure computation technique that allows evaluating computations on encrypted data and produces an encrypted result. HE is best suited when there is a client-server scenario with an untrusted server: the client simply has to encrypt its data and send it to the server, which performs the computations on the encrypted data and returns the encrypted result. The first HE schemes were partial, meaning that they either allowed the evaluation of additions or multiplications, but not both. Then, somewhat homomorphic schemes were developed, which allowed a limited number of both operations. One of the bottlenecks for obtaining an unlimited number of operations was the accumulation of noise introduced by one operation, which could eventually prevent the correct decryption. The first fully homomorphic encryption scheme that allowed the evaluation of both additions and multiplications on encrypted data was developed in~\cite{GentryPhD}, where, starting from a somewhat homomorphic encryption scheme, a bootstrapping operation was introduced. Bootstrapping allows to obliviously evaluate the scheme's decryption circuit and reduces the ciphertext noise. Other fully homomorphic encryption schemes include~\cite{Brakerski11,brakerski2012fully,fan2012somewhat,gentry2013homomorphic,brakerski2014leveled,cheon2017homomorphic}. For a thorough history and description of HE, see the survey~\cite{Martins2018survey}. 
Privacy solutions based on HE were proposed for genome matching, national security and critical infrastructure, healthcare databases, machine learning applications and control systems etc.~\cite{Archer2017applications, Aslett2015review, Riazi2018deep}. Of particular interest to us are the the works in control applications with HE, see~\cite{Gonzalez14,Kim16encrypting,Farokhi17,Darup18towards,Alexandru2018cloud,Murguia2018secure,Alexandru2019encrypted}, to name a few. Furthermore, there has been a soaring interest in homomorphically encrypted machine learning applications, from statistical analysis and data mining~\cite{Aslett2015review} to deep learning~\cite{Riazi2018deep}. 

Over the past decades, efforts in optimizing the computation and implementation of SMPC techniques, along with the improvement of network communication speed and more powerful hardware, have opened the way of SMPC deployment in \mbox{real-world} applications. Nevertheless, in the context of big data and the internet of things, which bring enormous amounts of data and large number of devices that want to participate in the computation, SMPC techniques lag behind plaintext computations due to the computational and communication bottleneck~\cite{Mirhoseini2016cryptoml,Mohassel17,Chen2018logistic}. 

The privacy definition for SMPC stipulates that the privacy of the inputs and intermediary results is ensured, but the output, which is a function of the inputs of all parties, is revealed. For applications such as smart meter aggregation~\cite{Acs2011have}, social media activity~\cite{Sala2011sharing}, health records~\cite{Dankar2012application}, deep learning~\cite{Abadi2016deep}, and where output privacy is valued over output accuracy, the SMPC privacy definition is not enough~\cite{Zhu2017differential} and needs to be augmented by guarantees of differential privacy. 

\textbf{Differential Privacy} (DP) refers to methods of concealing the information leakage from the result of the computation, even when having access to auxiliary information~\cite{Dwork06,Dwork08}. Intuitively, the contribution of the input of each individual to the output should be hidden from those who access the computation results. To achieve this, a carefully chosen noise is added to each entry such that the statistical properties of the database are preserved~\cite{Dwork2014algorithmic}, which introduces a trade-off between utility and privacy. 
When applied in a distributed system~\cite{Dwork2004privacy,Goyal2016distributed,Vadhan2017complexity}, DP provides a problematic solution for smaller number of entries in the dataset, as all parties would add noise that would completely drown the result of the computation. 
Several works combine SMPC with DP in order to achieve both computation privacy and output privacy, for instance~\cite{Rastogi2010differentially,Shi2011privacy,Pettai2015combining,Chase2017private,VadhanDPMPC2018}. 
An intuitive comparison between privacy-preserving centralized and decentralized computation approaches, with different privacy goals and utilities is given in Table~\ref{tab:comparison}, taken from~\cite{VadhanDPMPC2018}. 

\begin{table}[h]
\begin{center}
\vspace{-\topsep}
  \begin{tabular}{ p{0.25\textwidth} | p{0.2\textwidth} | p{0.29\textwidth} | p{0.21\textwidth} }
    \hline
    Model & Utility & Privacy & Who Holds the Data \\
    \hline \hline
    \raggedright Fully Homomorphic (or Functional) Encryption & \raggedright any desired query & \raggedright everything (except possibly the result of the query) & untrusted server \\ \hline
    \raggedright Secure Multi-Party Computation & \raggedright any desired query & \raggedright everything other than result of the query & original users or delegates \\ \hline
    \raggedright Centralized Differential Privacy & \raggedright statistical analysis of dataset & \raggedright individual-specific information & trusted curator \\ \hline
    \raggedright Multi-Party Differential Privacy & \raggedright statistical analysis of dataset & \raggedright individual-specific information & original users or delegates \\
    \hline
  \end{tabular}
  \caption{Comparison between centralized and multi-party privacy-preserving approaches~\cite{VadhanDPMPC2018}.}
  \vspace{-\topsep}
  \label{tab:comparison}
\end{center}
\end{table}
\vspace{-2\topsep}

In this chapter, we focus on proposing SMPC schemes for optimization and control problems, where we require the convergence and accuracy of the results. The scenario we consider will involve multiple clients and untrusted server(s). DP techniques can be further investigated in order to ensure output privacy, but we will not explore this avenue in the present chapter. 

\subsection{Dynamical data challenges}
\label{subsec:dynamical_data}
Cryptographic techniques were developed mainly for static data, such as databases, or independent messages. However, dynamical systems are iterative processes that generate structured and dependent data. Moreover, output data at one iteration/time step will often be an input to the computation at the next one. Hence, special attention is needed when using cryptographic techniques in solving optimization problems and implementing control schemes. For example, values encrypted with homomorphic encryption schemes will require ciphertext refreshing or bootstrapping if the multiplicative depth of the algorithm exceeds the multiplicative depth of the scheme; when using garbled circuits, a different circuit has to be generated for different iterations/time steps of the same algorithm; the controlled noise added for differential privacy at each iteration/time step will accumulate and drown the result etc. Furthermore, privacy is guaranteed as long as the keys and randomness are never reused, but freshly generated for each time step; then, the (possibly offline) phase in which uncorrelated randomness is generated, (independent from the actual inputs), has to be repeated for a continuously running process. In this chapter, we design the solution with all these issues in mind.

\subsection{Contribution and roadmap}
\label{subsec:roadmap}
Examples of control applications that require privacy include smart metering, crucial infrastructure control, prevention of industrial espionage, swarms of robots deployed in adversarial environments. In order to ensure the privacy of a control application, i.e., secure both the signals and the model, as well as the intermediate computations, most of the time, complex solutions have to be devised. 
We investigate a privacy-preserving cloud-based Model Predictive Control application, which, at a high level, requires privately computing additions, multiplications, comparisons and oblivious updates. Our solution encompasses several SMPC techniques: homomorphic encryption, secret sharing, oblivious transfer. The purpose of this chapter is to describe these cryptographic techniques and show how to combine them in order to build private cloud-based protocols that produce the desired control actions. 

The layout of the chapter is as follows: in Section~\ref{sec:MPC}, we describe the model predictive control problem that we address, we show how to compute the control action when there are no privacy requirements and we outline the privacy-preserving solution. In Section~\ref{sec:privacy_def}, we formally describe the adversarial model and privacy definition we seek for our multi-party computation problem. In Section~\ref{sec:cryptographic_tools}, we introduce the cryptographic tools that we will use in our solution and their properties. 
Then, in Section~\ref{sec:private_MPC}, we design the multi-party protocol for the model predictive control problem with encrypted states and model and prove its privacy. Finally, we discuss some details about the requirements of the proposed protocol in Section~\ref{sec:discussion}.

\subsection{Notation} 
We use bold-face lower case for vectors, e.g., $\mb x$, and bold-face upper case for matrices, e.g., $\mb A$. $\mbb N$ denotes the set of non-negative integers, $\mbb Z$ denotes the set of integers, $\mbb Z_N$ denotes the additive group of integers modulo $N$. For $n\in \mbb N$, let $[n]:=\{1,2,\ldots,n\}$. $\lambda$ denotes the security parameter and $1^\lambda$ is a standard cryptographic notation for the unary representation of $\lambda$, given as input to algorithms. 
$\E(x)$ and $[[x]]$ represent encryptions of the scalar value $x$. We use the same notation for multidimensional objects: an encryption of a matrix $\mb A$ is denoted by $\E(\mb A)$ (or $[[\mb A]]$) and is performed element-wise. 
A function $f:\mbb Z_{\geq q}\rightarrow \mbb R$ is called negligible if for all $c\in \mbb R_{> 0}$, there exists $n_0\in\mbb Z_{\geq 1}$ such that for all integers $n\geq n_0$, we have $|f(n)|\leq 1/n^c$. 
$\{0,1\}^\ast$ defines a sequence of bits of unspecified length.

\section{Model Predictive Control}
\label{sec:MPC}
We consider a discrete-time linear time-invariant system:
\begin{align}\label{eq:system}
\begin{split}
	\vec{x}(t+1) &= \vec{A} \vec{x}(t) + \vec{B} \vec{u}(t),\\
	\vec{x}(t) = \left[\begin{matrix} {\vec{x}^1(t)}^\intercal & \ldots & {\vec{x}^M(t)}^\intercal \end{matrix} \right], &\quad \vec{u}(t) = \left [ \begin{matrix} {\vec{u}^1(t)}^\intercal & \ldots & {\vec{u}^M(t)}^\intercal\end{matrix}\right],
\end{split}
\end{align}
with the state $\vec{x}\in\mc X\subseteq\mathbb R^{n}$ and the control input $\vec{u}\in\mc U\subseteq\mathbb R^{m}$. The system can be either centralized or partitioned in subsytems for $i\in[M]$, $\vec{x}^i(t) \in \mbb R^{n_i}$, $\sum_{i=1}^M n_i = n$ and $\vec{u}^i(t) \in \mbb R^{m_i}$, $\sum_{i=1}^M m_i = m$. The subsystems can be thought of as different agents or clients. We will address a privacy problem -- the agents desire to keep their local data private -- so it is more reasonable to consider partitions rather than a non-empty intersection of the states and control inputs. 

The Model Predictive Control (MPC) is the optimal control receding horizon problem with constraints written as:
\begin{align}\label{eq:mpc}
\begin{split}
	J_N^\ast(\vec{x}(t)) =& \min\limits_{\vec{u}_{0},\ldots,\vec{u}_{N-1}}\frac{1}{2} \left(\vec{x}_N^\intercal \vec{P} \vec{x}_N + \sum_{k=0}^{N-1} \vec{x}_k^\intercal \vec{Q} \vec{x}_k + \vec{u}_k^\intercal \vec{R} \vec{u}_k \right)\\
	s.t.&~ \vec{x}_{k+1} = \vec{A} \vec{x}_k + \vec{B} \vec{u}_k,~k=0,\ldots,N-1; ~~\vec{x}_0 = \vec{x}(t); \\
	&~\vec{u}_k\in \mc U,~k=0,\ldots,N-1,\\
\end{split}
\end{align}
where $N$ is the length of the horizon and $\vec{P},\vec{Q},\vec{R}\succ 0$ are cost matrices. 
We consider input constrained systems with box constraints $\vec{0}\in\mc U = \{\vec{l}_u\preceq \vec{u}\preceq \vec{h}_u\}$, and impose stability without a terminal state constraint, but with appropriately chosen costs and horizon, such that the \mbox{closed-loop} system has robust performance to bounded errors due to encryption. 
A survey on the conditions for stability of MPC is given in~\cite{Mayne00}. 

Through straightforward manipulations,~\eqref{eq:mpc} can be written as the quadratic problem~\eqref{eq:mpc(i)} -- see details on obtaining the matrices $\vec{H}$ and $\vec{F}$ in~\cite[Ch.~8,11]{MPCbook17} -- in the variable $\vec{U}:= \left[\vec{u}_0^\intercal \ \vec{u}_1^\intercal \ \ldots \ \vec{u}_{N-1}^\intercal \right]^\intercal$. 
For simplicity, we keep the same notation $\mc U$ for the augmented constraint set. 
After obtaining the optimal solution, the first $m$ components of $\vec{U}^\ast(\vec{x}(t))$ are applied as input to the system~\eqref{eq:system}: $\vec{u}^\ast(\vec{x}(t)) =\{\vec{U}^\ast(\vec{x}(t))\}_{1:m}$. 
\begin{equation}\label{eq:mpc(i)}
	\vec{U}^\ast(\vec{x}(t))~= \argmin\limits_{\vec{U}\in \mc U} \frac{1}{2} \vec{U}^\intercal \vec{H} \vec{U} + \vec{U}^\intercal \vec{F}^\intercal \vec{x}(t).
\end{equation}

\subsection{Solution without privacy requirements}
The privacy-absent cloud-MPC problem is depicted in Figure~\ref{fig:MPC}. The system~\eqref{eq:system} is composed of $M$ subsystems, that can be thought of as different agents, which measure their states and receive control actions, and of a setup entity which holds the system's model and parameters. The control decision problem is solved at the cloud level, by a cloud controller, which receives the system's model and parameters, the measurements, as well as the constraint sets imposed by each subsystem. The control inputs are then applied by one virtual actuator. Examples of such architecture include a smart building temperature control application, where the subsystems are apartments and the actuator is a machine in the basement, or the subsystems are robots in a swarm coordination application and the actuator is a ground control that sends them waypoints.

\begin{figure}[ht]
\centering
\vspace{-1.5\topsep}
    \includegraphics[width=0.72\textwidth]{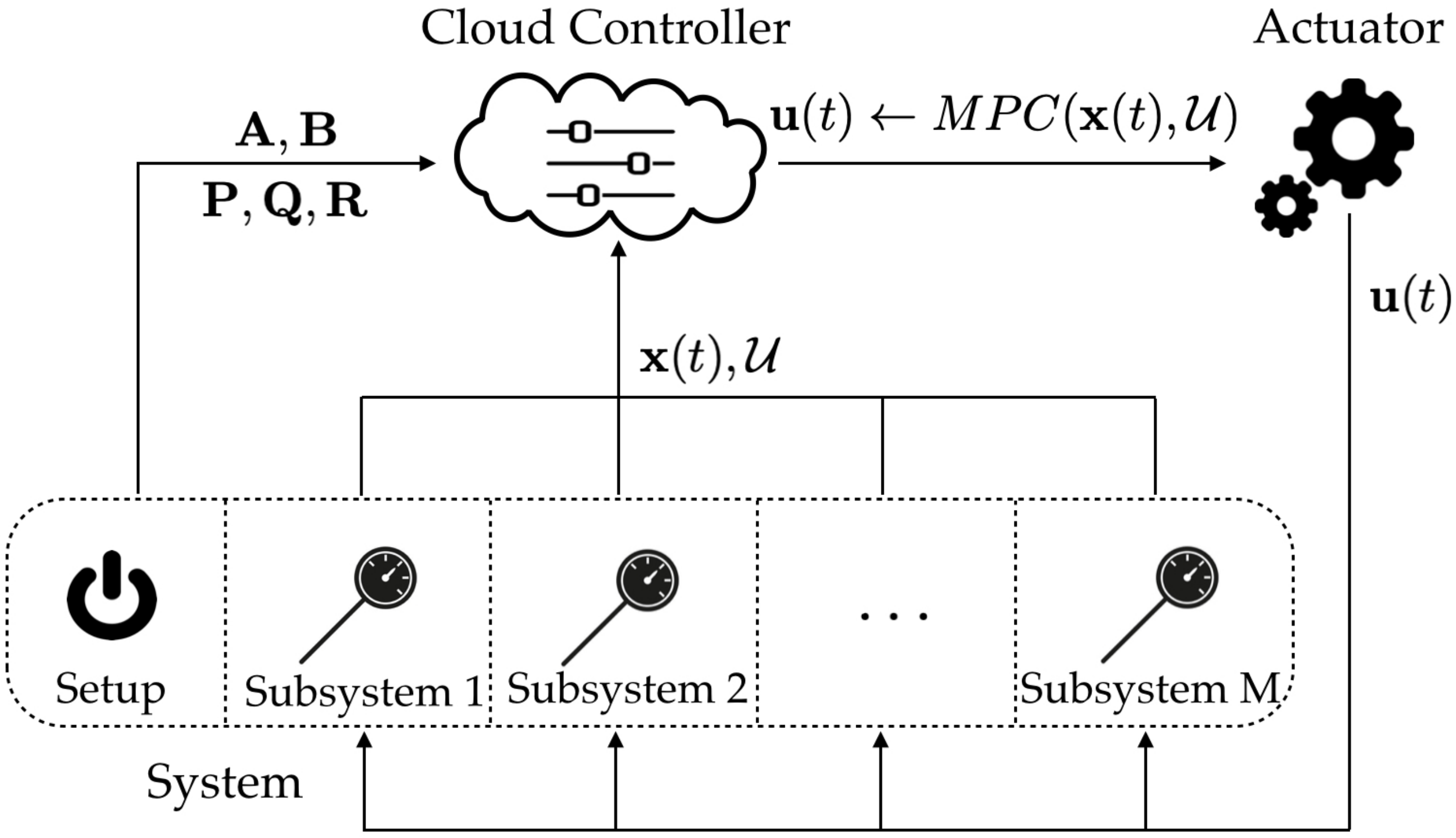}
  \caption{
  Cloud-based MPC problem for a system composed of a number of subsystems that measure their states and a setup entity which holds the system's model and parameters. The control action is computed by a cloud controller and sent to a virtual actuator.
  }
   \label{fig:MPC}
  \vspace{-\topsep} 
\end{figure}

The constraint set $\mc U$ is a hyperbox, so the projection step required for solving~\eqref{eq:mpc(i)} has a closed form solution, denoted by $\Pi_{\mc U}(\cdot)$ and the optimization problem can be efficiently solved with the projected Fast Gradient Method (FGM)~\cite{Nesterov13book}, given in~\eqref{eq:FGM}:%
\begin{subequations}\label{eq:FGM}
\begin{align}
&\text{For}~k=0\ldots,K-1 \notag\\
 	&\quad\quad \vec{t}_k \leftarrow \left(\vec{I}_{Mm} - \frac{1}{L}\vec{H}\right)\vec{z}_k - \frac{1}{L}\vec{F}^\intercal \vec{x}(t) \label{eq:iter_t}\\
	&\quad\quad \vec{U}_{k+1} \leftarrow \Pi_{\mc U}(\vec{t}_k) \label{eq:iter_U} \\
	&\quad\quad \vec{z}_{k+1} \leftarrow (1+\eta)\vec{U}_{k+1} - \eta \vec{U}_{k} \label{eq:iter_z}
\end{align}
\end{subequations}
where $\vec{z}_0 \leftarrow \vec{U}_0$. The objective function is strongly convex, since \mbox{$\vec{H}\succ 0$}, therefore we can use the constant step sizes \mbox{$L = \lambda_{max}(\vec{H})$} and \mbox{$\eta = (\sqrt{\kappa(\vec{H})}-1)/(\sqrt{\kappa(\vec{H})}+1)$}, where $\kappa(\vec{H})$ is the condition number of~$\vec{H}$. 
Warm starting can be used at subsequent time steps of the receding horizon problem by using part of the previous solution~$\vec{U}_K$ to construct a feasible initial iterate of the new optimization problem.

\subsection{Privacy objectives and overview of solution}

The system model and MPC costs $\vec{A},\vec{B}, \vec{P}, \vec{Q},\vec{R}$ are known only to the system, but not to the cloud and actuator, hence, the matrices $\vec{H}, \vec{F}$ in~\eqref{eq:mpc(i)} are also private. The measurements and constraints are not known to parties other than the corresponding subsystem and should remain private, such that the sensitive information of the subsystems is concealed. The control inputs should not be known by the cloud.

The goal of this work is to design a private cloud-outsourced version of the fast gradient method in~\eqref{eq:FGM} for the model predictive control problem, such that the actuator obtains the control action $\vec{u}^\ast(t)$ for system~\eqref{eq:system}, without learning anything else in the process and with only a minimum amount of computation. 
At the same time, the cloud controller should not learn anything other than what was known prior to the computation about the measurements $\vec{x}(t)$, the control inputs $\vec{u}^\ast(t)$, the constraints~$\mc U$, and the system model $\vec{H},\vec{F}$. 
We formally introduce the adversarial model and multi-party privacy definition in Section~\ref{sec:privacy_def}. 

As a primer to our private solution to the MPC problem, we briefly mention here the cryptographic tools used to achieve privacy. We will encrypt the data with a \textit{labeled homomorphic encryption} scheme, which allows us to evaluate an unlimited number of additions and one multiplication over encrypted data. The labeled homomorphic encryption builds on top of an \textit{additively homomorphic encryption} scheme, which allows only the evaluation of additions over encrypted data, a \textit{secret sharing} scheme, which enables the splitting of a message into two random shares, that cannot be used individually to retrieve the message, and a \textit{pseudorandom generator} that, given a key and a small seed, called \textit{label}, outputs a larger sequence of bits that is indistinguishable from random. The right choice of labels is essential for a seamless application of labeled homomorphic encryption on dynamical data, and we choose the labels to be the time steps at which the data is generated. These tools ensure that we can evaluate polynomials on the private data. Furthermore, the computations for determining the control action also involve projections on a feasible hyperbox. To achieve this in a private way, we make use of \textit{two-party private comparison} that involves exchanges of encrypted bits between two parties, and \textit{oblivious transfer}, that allows us to choose a value out of many values when the index is secret. These cryptographic tools will be described in detail in Section~\ref{sec:cryptographic_tools}, and our private cloud-based MPC solution that incorporates them will be presented in Section~\ref{sec:private_MPC}.

We present a standard framework in this chapter, but the same tools can be also applied on variations of the problem and architecture.

\section{Adversarial model and Privacy definition}
\label{sec:privacy_def}
In cloud applications, the service provider has to deliver the contracted service that was agreed upon, otherwise the clients switch to another service provider. This incentivizes the cloud to not alter the data it receives. 
The clients' interest is to obtain the correct result for the service they pay for, so we may assume they do not alter the data sent to the cloud. However, the parties can locally process copies of the data they receive in any fashion they want. 
This adversarial model is known as \mbox{semi-honest}, which is defined formally as follows:

\begin{definition}(\mbox{Semi-honest model}~\cite[Ch.~7]{Goldreich04foundationsII})
A party is \mbox{semi-honest} if it does not deviate from the steps of the protocol, but may store the transcript of the messages exchanged and its internal coin tosses, as well as process the data received in order to learn more information than stipulated by the protocol.
\end{definition}

This model also holds when considering eavesdroppers on the communication channels. Malicious and active adversaries -- that diverge from the protocols or tamper with the messages -- are not considered in this chapter. Privacy against malicious adversaries can be obtained by introducing zero-knowledge proofs and commitment schemes, at the cost of a computational overhead~\cite{Goldreich04foundationsII,Damgaard2010multiparty}.

We introduce some concepts necessary for the multi-party privacy definitions. An ensemble $X = \{X_\sigma\}_{\sigma\in\mathbb N}$ is a sequence of random variables ranging over strings of bits of length polynomial in $\sigma$, arising from distributions defined over a finite set~$\Omega$. 

\begin{definition}(Statistical indistinguishability~\cite[Ch.~3]{Goldreich03foundationsI})\label{def:stat_ind} 
~The ensembles $X= $ $\{X_\sigma\}_{\sigma\in\mathbb N}$ and $Y=\{Y_\sigma\}_{\sigma\in\mathbb N}$ are \textbf{statistically indistinguishable}, denoted $\stackrel{s}{\equiv}$, if for every positive polynomial $p$, and all sufficiently large $\sigma$:
\[\mr{SD} [X_\sigma,Y_\sigma]:=\frac{1}{2}\sum_{\alpha\in\Omega} \big|\Pr[X_\sigma = \alpha] - \Pr[Y_\sigma = \alpha] \big|< \frac{1}{p(\sigma)}.\]
The quantity on the left is called the statistical distance between the two ensembles.
\end{definition}

It can be proved that two ensembles are statistically indistinguishable if no algorithm can distinguish between them. Statistical indistinguishability holds against computationally unbounded adversaries. Computational indistinguishability is a weaker notion of the statistical version, as follows:

\begin{definition}\label{def:comp_ind}(Computational Indistinguishability~\cite[Ch.~3]{Goldreich03foundationsI}) 
The ensembles $X=\{X_\sigma\}_{\sigma\in\mathbb N}$ and $Y=\{Y_\sigma\}_{\sigma\in\mathbb N}$ are \textbf{computationally indistinguishable}, denoted $\stackrel{c}{\equiv}$, if for every probabilistic polynomial-time algorithm $D:\{0,1\}^\ast\rightarrow \{0,1\}$, called the distinguisher, every positive polynomial $p$, and all sufficiently large $\sigma$:
\[\big |\Pr_{x\leftarrow X_\sigma} [D(x) = 1] - \Pr_{y\leftarrow Y_\sigma}[D(y) = 1] \big| < \frac{1}{p(\sigma)}.\]
\end{definition}

Let us now look at the privacy definition that is considered in secure \mbox{multi-party} computation, that makes use of the real and ideal world paradigms. Consider a \mbox{multi-party} protocol $\Pi$ that executes a functionality $f = (f_1,\ldots,f_p)$ on inputs $I = (I_1,\ldots,I_p)$ and produces an output $f(I)=(f_1(I),\ldots,f_p(I))$ in the following way: the parties have their inputs, then exchange messages between themselves in order to obtain an output. If an adversary corrupts a party (or a set of parties) in this real world, after the execution of the protocol, it will have access to its (their) input, the messages received and its (their) output. In an ideal world, there is a trusted incorruptible party that takes the inputs from all the parties, computes the functionality on them, and then sends to each party its output. In this ideal world, an adversary that corrupts a party (or a set of parties) will only have access to its (their) input and output. This concept is illustrated for three parties in Figure~\ref{fig:real_ideal}. We then say that a multi-party protocol achieves computational privacy if, for each party, what a computationally bounded adversary finds out from the real world execution is equivalent to what it finds out from the ideal-world execution. 

\vspace{-\topsep}
\begin{figure}[ht]
\centering
\includegraphics[width=0.7\textwidth]{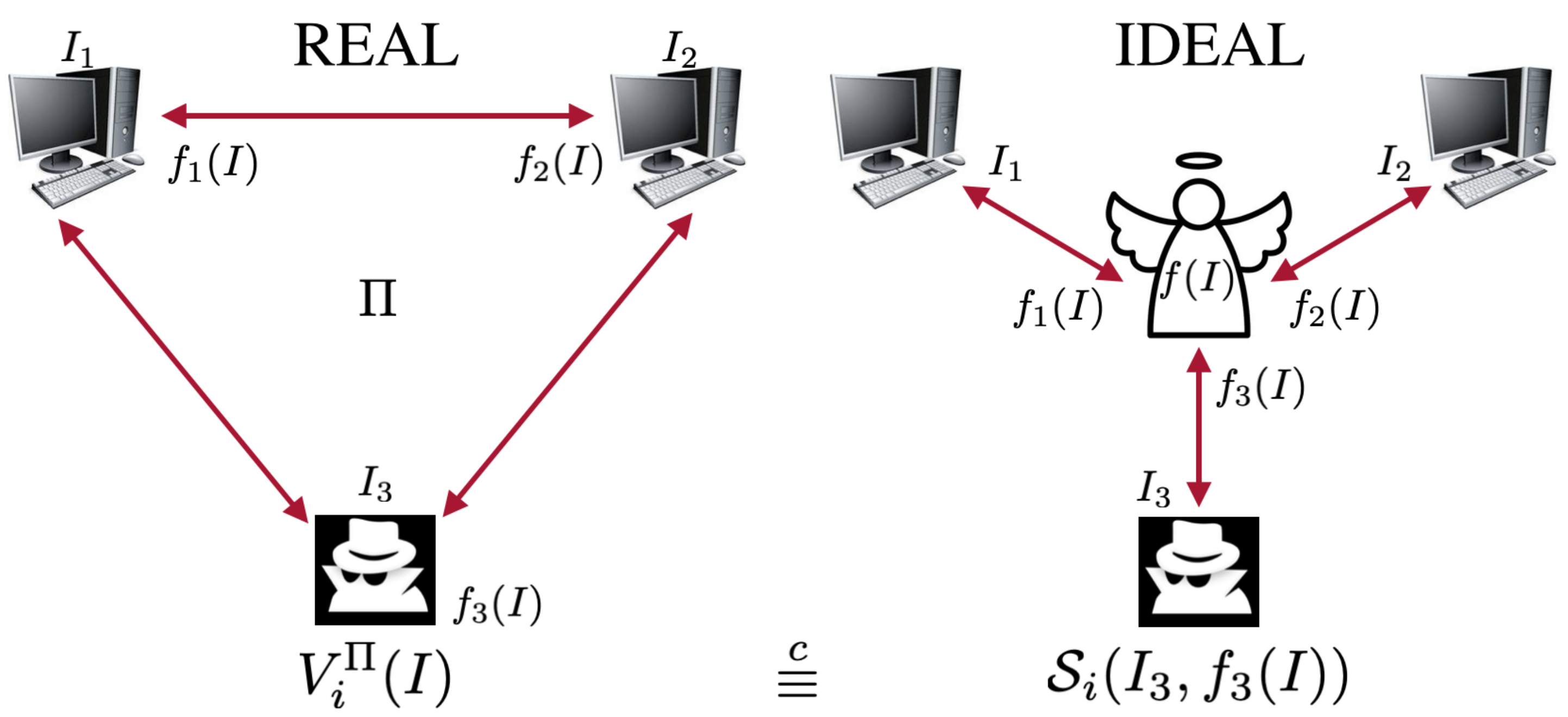}
\caption{Real and ideal paradigms for secure multi-party computation.}
\label{fig:real_ideal} 
\end{figure}
\vspace{-\topsep}

The formal definition is the following:
\begin{definition}\label{def:coalition_view}(Multi-party privacy w.r.t. \mbox{semi-honest} behavior~\cite[Ch.~7]{Goldreich04foundationsII}) 
 Let $f:(\{0,1\}^\ast)^p \rightarrow (\{0,1\}^\ast)^p$ be a $p$-ary functionality, where $f_i(x_1,\ldots,x_p)$ denotes the $i$-th element of $f(x_1,\ldots,x_p)$. Denote the inputs by $\bar x = (x_1,\ldots,x_p)$. For $I=\{i_1,\ldots,i_t\}\subset [p]$, we let $f_I(\bar x)$ denote the subsequence $f_{i_1}(\bar x),\ldots,$ $f_{i_t}(\bar x)$, which models a coalition of a number of parties. Let $\Pi$ be a $p$-party protocol that computes~$f$. The \textbf{view} of the $i$-th party during an execution of $\Pi$ on the inputs $\bar x$, denoted $V^\Pi_i (\bar x)$, is $(x_i,\mr{coins},m_1,\ldots,m_t)$, where $\mr{coins}$ represents the outcome of the $i$'th party's internal coin tosses, and $m_j$ represents the $j$-th message it has received. We let the view of a coalition be denoted by $V_I^\Pi (\bar x) = (I,V_{i_1}^\Pi(\bar x),\ldots,V_{i_t}^\Pi(\bar x))$. For a deterministic functionality $f$, we say that $\boldsymbol \Pi$ \textbf{privately computes} $\boldsymbol f$ if there exist simulators $S$, such that, for every $I\subset [p]$, it holds that, for $\bar x_t = (x_{i_1},\ldots,x_{i_t})$:
 \begin{align*}
\big\{S(I,\bar x_t,f_I(\bar x)) \big\}_{\bar x \in (\{0,1\}^\ast)^p} &\stackrel{c}{\equiv} 
\big\{V_I^\Pi (\bar x)\big\}_{\bar x \in (\{0,1\}^\ast)^p}.
\end{align*}
\end{definition}

This definition assumes the correctness of the protocol, i.e., the probability that the output of the parties is not equal to the result of the functionality applied to the inputs is negligible~\cite{Goldreich04foundationsII,Lindell17}. Auxiliary inputs, which are inputs that capture additional information available to each of the parties, (e.g., local configurations, side-information), are implicit in this definition~\cite[Ch.~7]{Goldreich04foundationsII}, \cite[Ch.~4]{Goldreich03foundationsI}.

Definition~\ref{def:coalition_view} states that the view of any of the parties participating in the protocol, on each possible set of inputs, can be simulated based only on its own input and output. 
For parties that have no assigned input and output, like the cloud server in our problem, Definition~\ref{def:coalition_view} captures the strongest desired privacy model.

\section{Cryptographic tools}
\label{sec:cryptographic_tools}

\subsection{Secret sharing}
\label{subsec:SS}
Secret sharing~\cite{Shamir1979share,Pedersen91} is a tool that distributes a private message to a number of parties, by splitting it into random shares. Then, the private message can be reconstructed only by an authorized subset of parties, which combine their shares. 

One common and simple scheme is the additive \mbox{2-out-of-2} secret sharing scheme, which involves a party splitting its secret message $m\in G$, where $G$ is a finite abelian group, into two shares, in the following way: generate uniformly at random an element $\mf{b}\in G$, subtract it from the message and then distribute the shares $\mf{b}$ and $m-\mf{b}$. This can be also thought of as a one-time pad~\cite{Vernam1926,Bellovin11} variant on $G$. Both shares are needed in order to recover the secret. 
The 2-out-of-2 secret sharing scheme achieves perfect secrecy, which means that the shares of two distinct messages are uniformly distributed on~$G$~\cite{Cramer2012secure}.

We will also use an additive blinding scheme weaker than secret sharing (necessary for the private comparison protocol in Section~\ref{subsec:comparison}): for messages of $l$ bits, $\mf{b}$ will be generated from a message space~$\mc M$ with length of $\lambda+l$ bits, where $\lambda$ is the security parameter, with the requirement that $\lambda + l $-bit messages can still be represented in $\mc M$, i.e., there is no wrap-around and overflow. 
The distribution of $m+\mf{b}$ is statistically indistinguishable from a random number sampled of $l+\lambda+1$ bits. 
Such a scheme is commonly employed for blinding messages, for instance in~\cite{Veugen10,Jeckmans13,Bost15}.

\vspace{-1.5\topsep}
\subsection{Pseudorandom generators}\label{subsec:PRG}
Pseudorandom generators are efficient deterministic functions that expand short seeds into longer pseudorandom bit sequences, that are computationally indistinguishable from truly random sequences. More details can be found in~\cite[Ch. 3]{Goldreich03foundationsI}.

\vspace{-1.5\topsep}
\subsection{Homomorphic Encryption}
\label{subsec:HE}
Let $\E(\cdot)$ denote a generic encryption primitive, with domain the space of private data, called \textbf{plaintexts}, and codomain the space of encrypted data, called \textbf{ciphertexts}. $\E(\cdot)$ also takes as input the public key, and probabilistic encryption primitives also take a random number. The decryption primitive $\D(\cdot)$ is defined on the space of ciphertexts and takes values on the space of plaintexts. $\D(\cdot)$ also takes as input the private key. \textbf{Additively homomorphic} schemes satisfy the property that there exists an operator~$\oplus$ defined on the space of ciphertexts such that: 
\begin{equation}~\label{eq:abstract_additive}
\E(a)\oplus \E(b) \subset \E(a+b),
\end{equation}
for any plaintexts $a,b$ supported by the scheme. We use set inclusion instead of equality because the encryption of a message is not unique in probabilistic cryptosystems. 
Intuitively, equation~\eqref{eq:abstract_additive} means that by performing this operation on the two encrypted messages, we obtain a ciphertext that is equivalent to the encryption of the sum of the two plaintexts. Formally, the decryption primitive $\D(\cdot)$ is a homomorphism between the group of ciphertexts with the operator $\oplus$ and the group of plaintexts with addition $+$, which justifies the name of the scheme. It is immediate to see that if a scheme supports addition between encrypted messages, it will also support subtraction, by adding the additive inverse, and multiplication between an integer plaintext and an encrypted message, obtained by adding the encrypted messages for the corresponding number of times. 

Furthermore, \textbf{multiplicatively homomorphic} schemes satisfy the property that there exists an operator $\otimes$ defined on the space of ciphertexts such that: 
\begin{equation}~\label{eq:abstract_multiplicative}
\E(a)\otimes \E(b) \subset \E(a\cdot b),
\end{equation}
for any plaintexts $a,b$ supported by the scheme. 
If the same scheme satisfies both~\eqref{eq:abstract_additive} and~\eqref{eq:abstract_multiplicative} for an unlimited amount of operations, it is called \textbf{fully homomorphic}. If a scheme satisfies both~\eqref{eq:abstract_additive} and~\eqref{eq:abstract_multiplicative} but only for a limited amount of operations, it is called \textbf{somewhat homomorphic}. 

\begin{remark}
	A homomorphic cryptosystem is malleable, which means that a party that does not have the private key can alter a ciphertext such that another valid ciphertext is obtained. Malleability is a desirable property in order to achieve \mbox{third-party} outsourced computation on encrypted data, but allows ciphertext attacks. 
In this work, we assume that the parties have access to authenticated channels, therefore an adversary cannot alter the messages sent by the honest parties.
\end{remark}

\subsubsection{Additively homomorphic cryptosystem}
\label{subsubsec:paillier}

Additively Homomorphic Encryptions schemes, abbreviated as $\mr{AHE}$, allow a party that only has encryptions of two messages to obtain an encryption of their sum and can be instantiated by various public key additively homomorphic encryption schemes such as~\cite{GM82,Paillier99,Damgaard01,Joye2013efficient}. 
Let \AHE=$(\hatKeyGen, \hatE,\hatD, \hat{\Add}, \hat{\cMult})$ be an instance of an asymmetric additively homomorphic encryption scheme, with $\mc M$ the message space and $\hatC$ the ciphertext space, where we will use the following abstract notation: $\hat\oplus$ denotes the addition on $\hatC$ and $\hat\otimes$ denotes the multiplication between a plaintext and a ciphertext. We save the notation without $\hat{(\cdot)}$ for the scheme in Section~\ref{subsubsec:LabHE}. Asymmetric or public key cryptosystems involve a pair of keys: a public key that is disseminated publicly, and which is used for the encryption of the private messages, and a private key which is known only to its owner, used for the decryption of the encrypted messages. We will denote the encryption of a message $m\in \mc M$ by $[[m]]$ as a shorthand notation for $\hatE(\mr{public~key},m)$.

\begin{enumerate}
\item $\hatKeyGen(1^\sigma)$: Takes the security parameter~$\sigma$ and outputs a public key $\hat{\mr{pk}}$ and a private key $\hat{\mr{sk}}$.

\item $\hatE(\hat{\mr{pk}},m)$: Takes the public key and a message $m\in\mc M$ and outputs a ciphertext $[[m]]\in\hatC$. 

\item $\hatD(\hat{\mr{sk}},c)$: Takes the private key and a ciphertext $c\in\hatC$ and outputs the message that was encrypted $m'\in\mc M$.

\item $\hat{\Add}(c_1,c_2)$: Takes ciphertexts $c_1,c_2\in\hatC$ and outputs ciphertext $c=c_1\hat\oplus c_2\in\hatC$ such that: 
$\hatD(\hat{\mr{sk}},c) = \hatD(\hat{\mr{sk}},c_1)+\hatD(\hat{\mr{sk}},c_2)$.

\item $\hat{\cMult}(m_1,c_2)$: Takes plaintext $m_1\in\mc M$ and ciphertext $c_2\in\hatC$ and outputs ciphertext $c=m_1\hat\otimes c_2\in\hatC$ such that:
$\hatD(\hat{\mr{sk}},c) = m_1\cdot\hatD(\hat{\mr{sk}},c_2)$.
\end{enumerate}

In this chapter, we use the popular Paillier encryption~\cite{Paillier99}, that has plaintext space $\mbb Z_N$ and ciphertext space $(\mbb Z_{N^2})^\ast$. Any other \AHE scheme that is semantically secure~\cite{GM82},~\cite[Ch. 4]{Goldreich04foundationsII} and circuit-private~\cite{Catalano2014boosting} can be employed. 

\subsubsection{Labeled homomorphic encryption}
\label{subsubsec:LabHE}

The model predictive control problem from Figure~\ref{fig:MPC}, and more general control decision problems as well, can be abstracted in the following general framework in Figure~\ref{fig:general}. Consider a cloud server that collects encrypted data from several clients. The data represents time series and is labeled with the corresponding time. A requester makes queries that can be written as multivariate polynomials over the data stored at the cloud server and solicits the result. 

\begin{figure}[ht]
  \centering
    \includegraphics[width=0.65\textwidth]{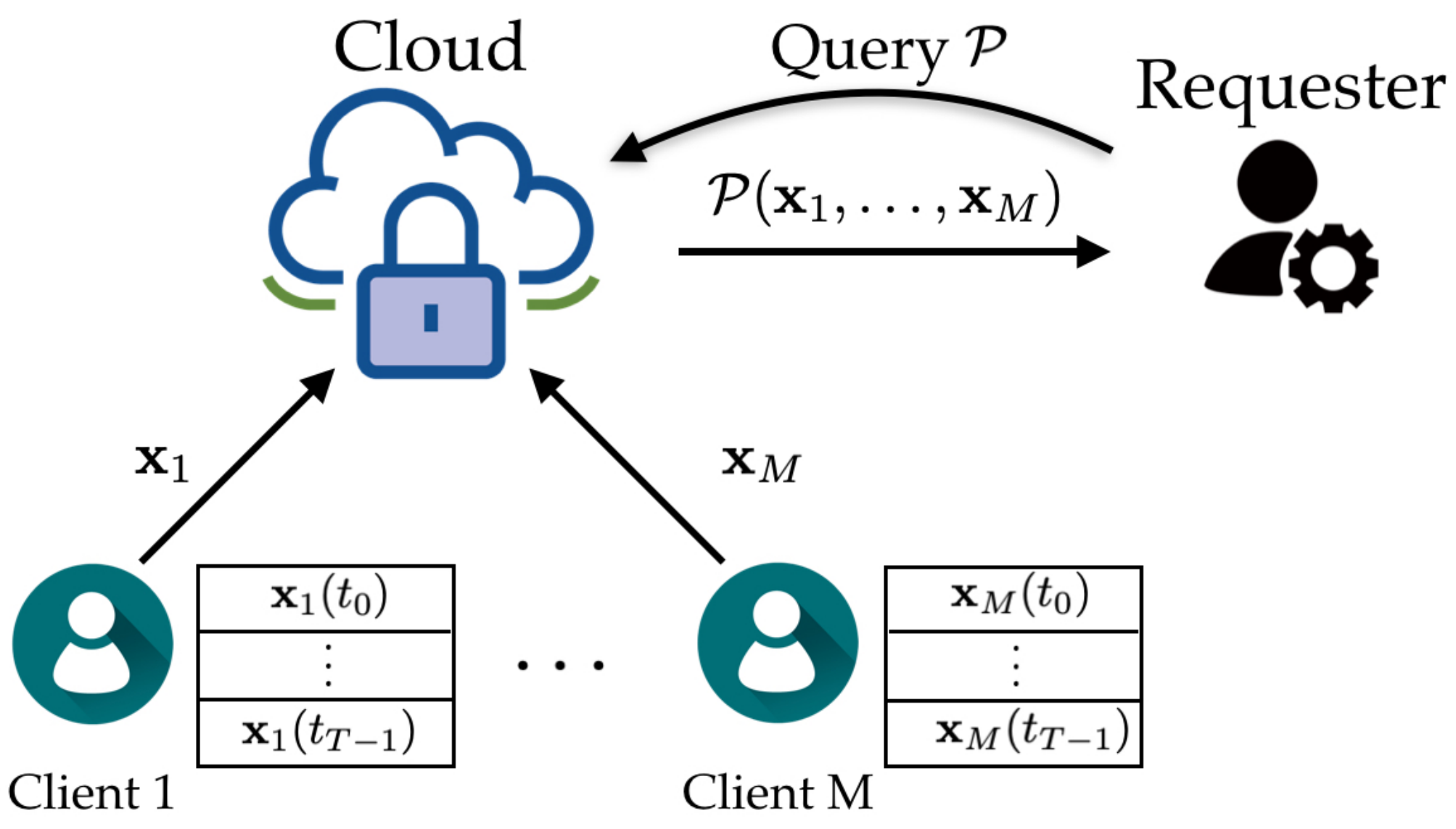}
  \caption{
  	The clients send their private data (collected over $T$ time steps or stored in a buffer) to a cloud server. A requester sends a query to the cloud, which evaluates it on the data and sends the result to the requester.
	}
\vspace{-\topsep}
   \label{fig:general}
\end{figure}

Labeled Homomorphic Encryption ($\mr{LabHE}$) can process data from multiple users with different private keys, as long as the requesting party has a master key. 
This scheme makes use of the fact that the decryptor (or requester in Figure~\ref{fig:general}) knows the query to be executed on the encrypted data, which we will refer to as a program. Furthermore, we want a cloud server that only has access to the encrypted data to be able to perform the program on the encrypted data and the decryptor to be able to decrypt the result. To this end, the inputs to the program need to be uniquely identified. Therefore, an encryptor (or client in Figure~\ref{fig:general}) assigns a unique label to each message and sends the encrypted data along with the corresponding encrypted labels to the server. Labels can be time instances, locations, id numbers etc. 

Denote by $\mc M$ the message space. An admissible function for \LabHE $f:\mc M^n \rightarrow \mc M$ is a multivariate polynomial of degree~2 on $n$ variables. A program that has labeled inputs is called a labeled program~\cite{Barbosa17labeled}:
\begin{definition}\label{def:labeled_prog}
A labeled program $\mc P$ is a tuple $(f,\tau_1,\allowbreak \ldots,\tau_n)$ where $f:\mc M^n \rightarrow \mc M$ is an admissible function on~$n$ variables and $\tau_i\in\{0,1\}^\ast$ is the label of the $i$-th input of $f$. Given $t$ programs $\mc P_1,\ldots,\mc P_t$ and an admissible function $g: \mc M^t$ $\rightarrow \mc M$, the composed program $\mc P^g$ is obtained~by evaluating $g$ on the outputs of $\mc P_1,\ldots,\mc P_t$, and can be denoted compactly as $\mc P^g = g(\mc P_1,\ldots,\mc P_t)$. The labeled inputs of $\mc P^g$ are all the distinct labeled inputs of $\mc P_1,\ldots,\mc P_t$. 
 \end{definition}
 
Let $f_{id}:\mc M \rightarrow \mc M$ be the identity function and $\mc \tau \in \{0,1\}^\ast$ be a label. Denote the identity program for input label $\tau$ by $\mc I_{\tau} = (f_{id},\tau)$. Any labeled program $\mc P = (f,\tau_1,\ldots,\tau_n)$, as in Definition~\ref{def:labeled_prog}, can be expressed as the composition of~$n$ identity programs $\mc P = f(\mc I_{\tau_1},\ldots,\mc I_{\tau_n})$.

\LabHE is constructed from an \AHE scheme with the requirement that the message space must be a public ring in which one can efficiently sample elements uniformly at random. The idea is that an encryptor splits their private message as described in Section~\ref{subsec:SS} into a random value (secret) and the difference between the message and the secret. For efficiency, instead of taking the secret to be a uniformly random value, we take it to be the output of a pseudorandom generator applied to the corresponding label. The label acts like the seed of the pseudo-random generator. The encryptor then forms the \LabHE ciphertext from the encryption of the first share along with the second share, yielding $\E(m)=(m - \mf{b},[[\mf{b}]])$, as described in Step~1 in the following. 
This enables us to compute one multiplication of two encrypted values and decrypt it, using the observation~\eqref{eq:LabHEhint}, as described in Steps~4 and 5. 
\begin{equation}\label{eq:LabHEhint}
m_1\cdot m_2 - \mf{b}_1\cdot \mf{b}_2 = (m_1 - \mf{b}_1)\cdot(m_2 - \mf{b}_2) + \mf{b}_1\cdot (m_2 - \mf{b}_2) + \mf{b}_2 \cdot (m_1 - \mf{b}_1). 
\end{equation}

Let $\mc M$ be the message space of the \AHE scheme, $\mc L \subset \{0,1\}^\ast$ denote a finite set of labels and $F:\{0,1\}^k \times \{0,1\}^\ast \rightarrow \mc M$ be a pseudorandom function that takes as inputs a key of size $k$ polynomial in $\sigma$ the security parameter, and a label from $\mc L$. Then \LabHE is defined as a tuple $\mr{LabHE}=(\Init,\KeyGen,\E,\Eval,\D)$:
\begin{itemize}
\item[1.] $\Init(1^\sigma)$: Takes the security parameter~$\sigma$ and outputs master secret key~$\mr{msk}$ and master public key~$\mr{mpk}$ for $\mathrm{AHE}$. 

\item[2.] $\KeyGen(\mr{mpk})$: Takes the master public key~$\mr{mpk}$ and outputs for each user~$i$ a user secret key~$\mr{usk_i}$ and a user public key~$\mr{upk_i}$.

\item[3.] $\E(\mr{mpk,upk},\tau,m)$: Takes the master public key, a user public key, a label $\tau\in \mc L$ and a message $m\in\mc M$ and outputs a ciphertext $C = (a,\beta)$. It is composed of an online and offline part:
	\item $\text{Off-E}(\mr{usk},\tau)$: Computes the secret $\mf{b}\leftarrow F(usk,\tau)$ and outputs $\mr{C_{off}} = (\mf{b},[[\mf{b}]])$.
	\item $\text{On-E}(\mr{C_{off}},m)$: Outputs $C=(m - \mf{b},[[\mf{b}]]) =: (a,\beta) \in \mc M\times \hat{\mc C}.$

\item[4.] $\Eval(\mr{mpk},f,C_1,\ldots,C_t)$: Takes the master public key, an admissible function $f:\mc M^t \rightarrow \mc M$, $t$ ciphertexts and returns a ciphertext $C$.  
$\Eval$ is composed of the following building blocks:
	\item $\Mult(C_1,C_2)$: Takes $C_i = (a_i,\beta_i)\in \mc M\times \hatC$ for $i=1,2$ and outputs \break$C = [[a_1\cdot a_2]]\hat\oplus (a_1\hat\otimes \beta_2) \hat\oplus (a_2\hat\otimes \beta_1) = [[m_1\cdot m_2 - \mf{b}_1\cdot \mf{b}_2]] =: \alpha \in \hatC$.
	\item $\Add(C_1,C_2)$: If $C_i = (a_i,\beta_i)\in\mc M \times \hatC$ for $i=1,2$, then outputs $C = (a_1+a_2,$ $\beta_1\hat\oplus \beta_2) =: (a,\beta) \in \mc M \times \hatC$. If both $C_i = \alpha_i \in \hatC$, for $i=1,2$, then outputs $C = \alpha_1\hat\oplus\alpha_2 =: \alpha \in \hatC$. If $C_1 = (a_1,\beta_1) \in\mc M \times \hatC$ and $C_2 = \alpha_2 \in \hatC$, then outputs $C = (a_1,\beta_1\hat\oplus \alpha_2) =: (a,\beta) \in \mc M\times \hatC$.
	\item $\cMult(c,C')$: Takes a plaintext $c\in\mc M$ and a ciphertext $C'$. If $C' = (a',\beta')\in\mc M\times \hatC$, outputs $C = (c \cdot a',c\hat\otimes \beta') =: (a,\beta)$ $\in \mc M\times \hatC$. If $C' = \alpha'\in\hatC$, outputs $C = c\hat\otimes \alpha' =: \alpha \in \hatC$.

\item[5.] $\D(\mr{msk,usk}_{1,\ldots,t},\mc P,C)$: Takes the master secret key, a vector of $t$ user secret keys, a labeled program $\mc P$ and a ciphertext $C$. It has an online and an offline part:
	\item $\text{Off-D}(\mr{msk},\mc P)$: Parses $\mc P$ as $(f,\tau_1,\ldots,\tau_t)$. For $i\in[t]$, it computes the secrets $\mf{b}_i = F(\mr{usk}_i,\tau_i)$, $\mf{b}=f(\mf{b}_1,\ldots,\mf{b}_t)$ and outputs $\mr{msk_{\mc P}(msk,\mf{b})}$.
	\item $\text{On-D}(\mr{sk_{\mc P}},C)$: If $C = (a,\beta)\in\mc M \times \hatC$: either output (i) $m = a+\mf{b}$ or (ii) output $m = a + \hatD (\mr{msk},\beta)$. If $C \in \hatC$, output $m = \hatD (\mr{msk},C)+\mf{b}$.
\end{itemize}

The cost of an online encryption is the cost of an addition in~$\mc M$. The cost of online decryption is independent of $\mc P$ and only depends on the complexity of $\hatD$. 

Semantic security characterizes the security of a cryptosystem. The definition of semantic security is sometimes given as a cryptographic game~\cite{GM82}. 

\begin{definition}\label{def:semantic}(Semantic Security~\cite[Ch.~5]{Goldreich04foundationsII})
An encryption scheme is \textbf{semantically secure} if for every probabilistic polynomial-time algorithm, $\mc A$, there exists a probabilistic polynomial-time algorithm $\mc A'$ such that for every two polynomially bounded functions $f,h:\{0,1\}^\ast\rightarrow \{0,1\}^\ast$, for any probability ensemble $\{X_\sigma\}_{\sigma\in\mathbb N}$ of length polynomial in~$\sigma$, for any positive polynomial $p$ and sufficiently large $\sigma$: 
\begin{align*}
\Pr \left[ \mc A(\E(X_\sigma),h(X_\sigma),1^\sigma) = f(X_\sigma)\right] < \Pr \left[ \mc A'(h(X_\sigma),1^\sigma) = f(X_\sigma)\right] + \frac{1}{p(\sigma)},
\end{align*}
The probability is taken over the ensemble $X_\sigma$ and the internal coin tosses in $\mc A,\mc A'$.
\end{definition}

Under the assumption of decisional composite residuosity~\cite{Paillier99}, the Paillier scheme is semantically secure and has indistinguishable encryptions. 
Moreover, the \LabHE scheme satisfies semantic security given that the underlying homomorphic encryption scheme is semantically secure and the function~$F$ is pseudorandom.

In~\cite{Barbosa17labeled} it is proved that \LabHE also satisfies \mbox{context-hiding} (decrypting the ciphertext does not reveal anything about the inputs of the computed function, only the result of the function on those inputs).

\subsection{Oblivious transfer}
\label{subsec:OT}
Oblivious transfer is a technique used when one party wants to obtain a secret from a set of secrets held by another party~\cite[Ch.~7]{Goldreich04foundationsII}. Party A has $k$ secrets $(\sigma_0,\ldots, \sigma_{k-1})$ and party B has an index $i\in\{0,\ldots,k-1\}$. The goal of A is to transmit the $i$-th secret requested by the receiver without knowing the value of the index~$i$, while B does not learn anything other than $\sigma_i$. This is called 1-out-of-$k$ oblivious transfer. There are many constructions of oblivious transfer that achieve security as in the two-party version of the simulation definition (Definition~\ref{def:coalition_view}). Many improvements in efficiency, e.g., precomputation, and security have been proposed, see e.g.,~\cite{Ishai2008founding,Nielsen2012new}.

We will use the standard 1-out-of-2 oblivious transfer, where the inputs of party A are $[[\sigma_0]],[[\sigma_1]]$ and party B holds $i\in\{0,1\}$ and the secret key and has to obtain $\sigma_i$. We will denote this by $\sigma_i\leftarrow \mr{OT}([[\sigma_0]],[[\sigma_1]],i,\hat{\mr{sk}})$. We will also use a variant where party A has to obliviously obtain the AHE-encrypted $[[\sigma_i]]$, and A has $[[\sigma_0]],[[\sigma_1]]$ and party B holds $i$, for $i\in\{0,1\}$, and the secret key. We will denote this variant by $[[\sigma_i]]\leftarrow \mr{OT'}([[\sigma_0]],[[\sigma_1]],i,\hat{\mr{sk}})$. 

The way the variant $\mr{OT}'$ works is that A chooses at random $r_0,r_1$ from the message space $\mc M$, and sends shares of the messages to B: $[[v_0]] := \hat\Add([[\sigma_0]], [[r_0]]), [[v_1]] :=\hat\Add([[\sigma_1]], [[r_1]])$. 
B selects $v_i$ and sends back to A the encryption of the index~$i$: $[[i]]$ and $\hat\Add([[v_i]],[[0]])$, such that A cannot obtain information about $i$ by comparing the value it received with the values it sent. Then, A computes:
\[
[[\sigma_i]] = \hat\Add\left([[v_i]], \hat\cMult(r_0,\hat\Add([[i]],[[-1]])),\hat\cMult(-r_1,[[i]])\right).
\]

\begin{proposition}\label{prop:OT}
$[[\sigma_i]]\leftarrow \mr{OT'}([[\sigma_0]],[[\sigma_1]],i,\hat{\mr{sk}})$ is private w.r.t. Definition~\ref{def:coalition_view}.
\end{proposition}

We will show in the following proof how to construct the simulators in Definition~\ref{def:coalition_view} in order to prove the privacy of the oblivious transfer variant we use.

\begin{proof}
Let us construct the view of A, with inputs $\hat{\mr{pk}}, [[\sigma_0]],[[\sigma_1]]$ and output~$[[\sigma_i]]$:
\[V_A(\hat{\mr{pk}}, [[\sigma_0]],[[\sigma_1]]) = \left(\hat{\mr{pk}}, [[\sigma_0]],[[\sigma_1]], r_0, r_1, [[i]], [[v_i]], [[\sigma_i]], \mr{coins} \right),\]
where $\mr{coins}$ are the random values used for encrypting $r_0$ and $r_1$ and $[[-1]]$.
The view of party B, that has inputs $i,\hat{\mr{pk}},\hat{\mr{sk}}$ and no output, is:
\[V_B(i,\hat{\mr{pk}},\hat{\mr{sk}}) = \big(i,\hat{\mr{pk}},\hat{\mr{sk}}, [[v_0]], [[v_1]], \mr{coins}\big),\]
where $\mr{coins}$ are the random values used for encrypting $v_i$, $i$ and $0$.

Now let us construct a simulator $S_A$ that generates an indistinguishable view from party A. $S_A$ takes as inputs $\hat{\mr{pk}}, [[\sigma_0]], [[\sigma_1]], [[\sigma_i]]$ and generates $\widetilde{r_0}$ and $\widetilde{r_1}$ as random values in $\mc M$. It then selects a random bit $\widetilde{i}$ and encrypts it with $\hat{\mr{pk}}$ and computes $[[\widetilde{v_i}]] = \hat\Add\big([[\sigma_i]], \hat\cMult(-\widetilde{r_0},\hat\Add([[\widetilde{i}]],[[-1]])),\hat\cMult(\widetilde{r_1},[[\widetilde{i}]])\big)$. It also generates $\widetilde{\mr{coins}}$ for three encryptions. $S_A$ outputs:
\[S_A(\hat{\mr{pk}}, [[\sigma_0]],[[\sigma_1]],[[\sigma_i]]) = \big(\hat{\mr{pk}}, [[\sigma_0]],[[\sigma_1]], \widetilde{r_0}, \widetilde{r_1}, [[\widetilde{i}]], [[\widetilde{v_i}]], [[\sigma_i]], \widetilde{\mr{coins}} \big).\]

First, $\widetilde{r_0}, \widetilde{r_1}$ and $\widetilde{\mr{coins}}$ are statistically indistinguishable from $r_0, r_1$ and $ \mr{coins}$ because they were generated from the same distributions. Second, $[[\widetilde i]]$ and $[[\widetilde{v_i}]]$ are indistinguishable from $[[i]]$ and $[[v_i]]$ because \AHE is semantically secure and has indistinguishable encryptions, and because $[[\sigma_i]]$ is a refreshed value of $[[\sigma_0]]$ or $[[\sigma_1]]$. This means A cannot learn any information about $i$, hence $[[i]]$ looks like the encryption of a random bit, i.e., like $[[\widetilde{i}]]$. 
Thus, $V_A(\hat{\mr{pk}}, [[\sigma_0]],[[\sigma_1]]) \stackrel{c}{\equiv} S_A(\hat{\mr{pk}}, [[\sigma_0]],[[\sigma_1]],[[\sigma_i]])$.

A simulator $S_B$ for party B takes as inputs $i,\hat{\mr{pk}},\hat{\mr{sk}}$ and generates two random values from $\mc M$, names them $\widetilde{v_0}$ and $\widetilde{v_1}$ and encrypts them. It then generates $\widetilde{\mr{coins}}$ as random values for three encryptions. $S_B$ outputs:
\[S_B(i,\hat{\mr{pk}},\hat{\mr{sk}}) = \left(i,\hat{\mr{pk}},\hat{\mr{sk}}, [[\widetilde{v_0}]], [[\widetilde{v_1}]], \widetilde{\mr{coins}}\right).\]

First, $\widetilde{\mr{coins}}$ are statistically indistinguishable from $ \mr{coins}$ because they were generated from the same distribution. Second, $\widetilde{v_0}$ and $\widetilde{v_1}$ are also statistically indistinguishable from each other and from $v_0$ and $v_1$ due to the security of the one-time pad. Their encryptions will also be indistinguishable. Thus, $V_B(i,\hat{\mr{pk}},\hat{\mr{sk}}) \stackrel{c}{\equiv} S_B(i,\hat{\mr{pk}},\hat{\mr{sk}})$.
\end{proof}

\subsection{Private comparison}
\label{subsec:comparison}
Consider a \mbox{two-party} computation problem of two inputs encrypted under an encryption scheme that does not preserve the order from plaintexts to ciphertexts. A survey on the state of the art of private comparison protocols on private inputs owned by two parties is given in~\cite{Couteau16}. 
In~\cite{DGK07,DGK09correction}, Damg{\aa}rd, Geisler and Kr{\o}igaard describe a protocol for secure comparison and, towards that functionality, they propose the DGK additively homomorphic encryption scheme with the property that, given a ciphertext and the secret key, it is efficient to determine if the encrypted value is zero without fully decrypting it. This is useful when making decisions on bits. The authors also prove the semantic security of the DGK cryptosystem under the hardness of factoring assumption. We denote the DGK encryption of a scalar value by $[\cdot]$ and use the same operations notation $\hat\oplus$, $\hat\ominus$ and $\hat\otimes$ to abstract the addition and subtraction between and encrypted values, respectively, the multiplication between a plaintext and an encrypted values. 

Consider two parties A and B, each having a private value $\alpha$ and $\beta$. Using the binary representations of $\alpha$ and $\beta$, the two parties exchange $l$ blinded and encrypted values such that each of the parties will obtain a bit $\delta_A\in\{0,1\}$ and $\delta_B\in\{0,1\}$ that satisfy the following relation: $\delta_A\veebar\delta_{B} = (\alpha\leq \beta)$, after executing Protocol~\ref{prot:plain_DGKV}, where $\veebar$ denotes the exclusive or operation. The protocol is described in~\cite[Protocol 3]{Veugen12}, where an improvement of the DGK scheme is proposed. By applying some extra steps, as in~\cite[Protocol 2]{Veugen12}, one can obtain a protocol for private \mbox{two-party} comparison where party A has two encrypted inputs with an \AHE scheme $[[a]],[[b]]$, with $a,b$ represented on $l$ bits, described in Protocol~\ref{prot:encr_DGKV}.

\begin{protocol}[Private \mbox{two-party} comparison with encrypted inputs using DGK] 
\label{prot:encr_DGKV}
\small
\begin{algorithmic}[1]
 \Require{A: $[[a]],[[b]]$; B: $\hat{\mr{sk}}, sk_{DGK}$}
 \Ensure{B: $(\delta=1)\equiv(a\leq b)$}
\State A: choose uniformly at random $r$ of $l+1+\lambda$ bits, compute $[[z]] \leftarrow [[b]]\hat\ominus[[a]] \hat\oplus [[2^l + r]]$ and send it to B. Then compute $\alpha \leftarrow r\bmod 2^l$.
\State B: decrypt $[[z]]$ and compute $\beta \leftarrow z\bmod 2^l$.
\State A,B: execute Protocol~\ref{prot:plain_DGKV}.
\State B: send $[[z\div 2^l]]$ and $[[\delta_B]]$ to A.
\State A: $[[(\beta < \alpha)]]\leftarrow [[\delta_B]]$ if $\delta_A = 1$ and $[[(\beta < \alpha)]]\leftarrow [[1]]\hat\ominus [[\delta_B]]$ otherwise.
\State A: compute $[[\delta]] \leftarrow [[z\div 2^l]] \hat\ominus ([[r\div 2^l]]\hat\oplus [[(\beta < \alpha)]])$ and send it to B.
\State B: decrypts $\delta$.
\end{algorithmic}
\end{protocol}

\begin{protocol}[Private \mbox{two-party} comparison with plaintext inputs using DGK] 
\label{prot:plain_DGKV}
\small
\begin{algorithmic}[1]
 \Require{A: $\alpha$; B: $\beta,sk_{DGK}$}
 \Ensure{A: $ \delta_A$; B: $\delta_B$ such that $\delta_A\veebar\delta_B = (\alpha\leq\beta)$}
\State B: send the encrypted bits $[\beta_i]$, $0\leq i <l$ to A.
\For{each $0\leq i <l$}
\State A: $[\alpha_i\veebar \beta_i]\leftarrow [\beta_i]$ if $\alpha_i = 0$ and $[\alpha_i\veebar \beta_i]\leftarrow[1]\hat\oplus(-1)\hat\otimes[\beta_i]$ otherwise.
\EndFor
\State A: Choose a uniformly random bit $\delta_A\in\{0,1\}$.
\State A: Compute the set $\mc L = \{i| 0\leq i <l \text{ and } \alpha_i = \delta_A \}$.
\For{each $i\in\mc L$}
\State A: compute $[c_i]\leftarrow [\alpha_{i+1}\veebar\beta_{i+1}] \hat\oplus \ldots \hat\oplus [\alpha_{l}\veebar\beta_{l}])$ .
\State A: $[c_i]\leftarrow [1]\hat\oplus[c_i]\oplus(-1)\hat\otimes[\beta_i]$ if $\delta_A = 0$ and $[c_i]\leftarrow  [1]\hat\oplus[c_i]$ otherwise.
\EndFor
\State A: generate uniformly random non-zero values $r_i$ of $2t$ bits (see~\cite{DGK09correction}), $0\leq i <l$.
\For{each $0\leq i <l$}
\State A: $[c_i]\leftarrow r_i\hat\otimes[c_i] $ if $ i\in\mc L $ and $ [c_i]\leftarrow [r_i] $ otherwise.
\EndFor
\State A: send the values $[c_i]$ in random order to B.
\State B: if at least one of the values $c_i$ is decrypted to zero, set $\delta_B \leftarrow 1$, otherwise set $\delta_B\leftarrow 0$.
\end{algorithmic}
\end{protocol}

\begin{proposition}\label{prop:DGK}(~\cite{DGK07,Veugen12})
	Protocol~\ref{prot:encr_DGKV} is private w.r.t. Definition~\ref{def:coalition_view}.
\end{proposition}

\section{MPC with encrypted model and encrypted signals}\label{sec:private_MPC}
As remarked in the Introduction, in many situations it is important to protect not only the signals (e.g., the states, measurements), but also the system model. To this end, we propose a solution that uses Labeled Homomorphic Encryption to achieve encrypted multiplications and the private execution of MPC on encrypted data. $\mr{LabHE}$ has a useful property called unbalanced efficiency that can be observed from Section~\ref{subsubsec:LabHE} and was described in~\cite{Catalano2015using}, which states that only the evaluator is required to perform operations on ciphertexts, while the decryptor performs computations only on the plaintext messages. 
We will employ this property by having the cloud perform the more complex operations and the actuator the more efficient ones.

Our protocols will consist of three phases: the offline phase, in which the computations that are independent from the specific data of the users are performed, the initialization phase, in which the computations related to the constant parameters in the problem are performed, and the online phase, in which computations on the variables of the problem are performed. The initialization phase can be offline, if the constant parameters are a priori known, or online otherwise.

Figure~\ref{fig:encrypted_MPC} represents the private version of the MPC diagram from Figure~\ref{fig:MPC}, where the quantities will be briefly described next and more in detail in Section~\ref{subsec:protocol}. The actuator holds the MPC query functionality, denoted by $f_{\mr{MPC}}$. 
Offline, the actuator generates a pair of master keys, as described in Section~\ref{subsubsec:LabHE} and publishes the master public key. The setup and subsystems generate their secret keys and send the corresponding public keys to the actuator. Still offline, these parties generate the labels corresponding to their data with respect to the time stamp and the size of the data. As explained in Section~\ref{subsubsec:LabHE}, the labels are crucial to achieving the encrypted multiplications. Moreover, when generating them, it is important to make sure that no two labels that will be encrypted with the same key are the same. When the private data are times series, as in our problem, the labels can be easily generated using the time steps and sizes corresponding to each signal, with no other complex synchronization process necessary between the actors. This is shown in Protocol~\ref{prot:prior_MPC}. 

The setup entity sends the $\mr{LabHE}$ encryptions of the state matrices and costs to the cloud controller  before the execution begins. The subsystems send the encryptions of the input constraints to the cloud controller, also before the execution begins. Online, at every time step, the subsystems encrypt their measurements and send them to the cloud. After the cloud performs the encrypted MPC query for one time step, it sends the encrypted control input at the current time step to the actuator, which decrypts it and inputs it to the system. In Protocol~\ref{prot:iter_MPC}, we describe how the encrypted MPC query is performed by the parties, which involves the actuator sending an encryption of the processed result back to the cloud such that the computation can continue in the future time steps.

\vspace{-\topsep}
\begin{figure}[ht!]
  \centering
    \includegraphics[width=0.73\textwidth]{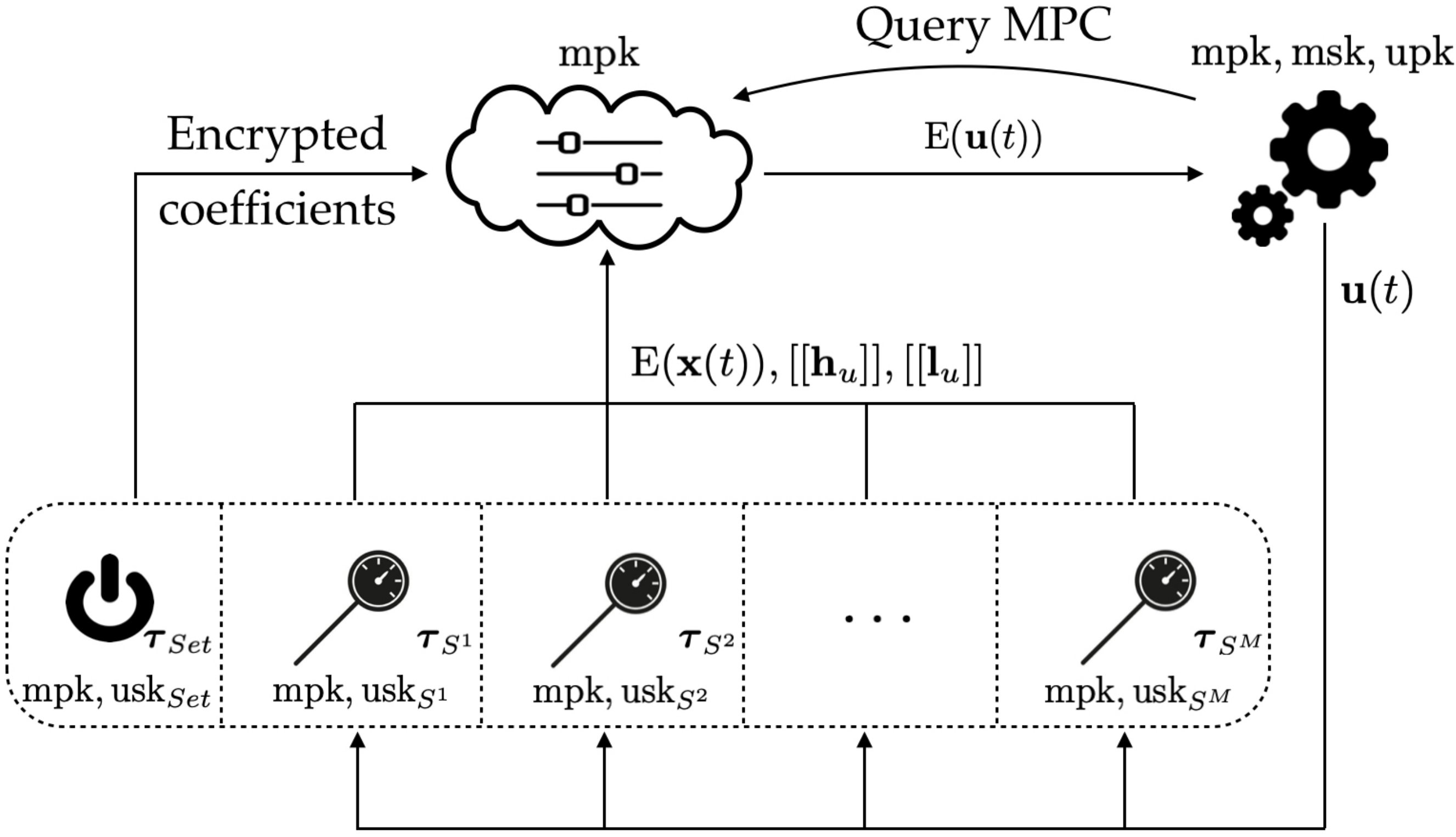}
  \caption{
  	The setup and subsystems send their encrypted data to the cloud. The cloud has to run the MPC algorithm on the private measurements and the system's private matrices and send the encrypted result to the actuator. The latter then actuates the system with the decrypted inputs.}
   \label{fig:encrypted_MPC}
\end{figure}

We will now show how to transform the FGM in Equation~\eqref{eq:FGM} into a private version, using Labeled Homomorphic Encryption and private two-party comparison. 
The message space $\mc M$ of the encryption schemes we use is $\mbb Z_N$, the additive group of integers modulo a large value $N$. This means that, prior to encryption, the values have to be represented as integers. For now, assume that this preprocessing step has been already performed. We postpone the details to Section~\ref{subsec:error_analysis}.

First, let us write $\vec{t}_k$ in~\eqref{eq:iter_t} as a function of $\vec{U}_{k}$ and $\vec{U}_{k-1}$:
\vspace{-0.5\topsep}
\begin{align*}
\vec{t}_k &= \left(\vec{I}_{Mm} - \frac{1}{L}\vec{H}\right)z_k - \frac{1}{L}\vec{F}^\intercal \vec{x}(t)= \left(\vec{I}_{Mm} - \frac{1}{L}\vec{H}\right) \left[(1+\eta)\vec{U}_{k} - \eta \vec{U}_{k-1} \right] - \frac{1}{L}\vec{F}^\intercal \vec{x}(t)\\
      &= \vec{U}_k + \eta(\vec{U}_k - \vec{U}_{k-1}) - \frac{1}{L}\vec{H} \vec{U}_k - \frac{\eta}{L}\vec{H} (\vec{U}_k - \vec{U}_{k-1}) - \frac{1}{L}\vec{F}^\intercal \vec{x}(t).
\end{align*}
If we consider the composite variables $\frac{1}{L}\vec{H}$, $\frac{\eta}{L}\vec{H}$, $\frac{1}{L}\vec{F}$ and variables $\vec{U}_k, \vec{U}_{k-1}, \vec{x}(t)$, then $\vec{t}_k$ can be written as a degree-two multivariate polynomial. This allows us to compute $[[\vec{t}_k]]$ using $\mr{LabHE}$. 

Then, one encrypted iteration of the FGM, where we assume that the cloud has access to $\E\left(-\frac{1}{L}\vec{H}\right)$, $\E\left(-\frac{\eta}{L}\vec{H}\right)$, $\E\left(\frac{1}{L}\vec{F}^\intercal\right)$, $\E(\vec{x}(t))$, $[[\vec{h}_u]], [[\vec{l}_u]]$ can be written as follows. Denote the computation on the inputs mentioned previously as $f_\mr{iter}$. We use both $\Add$ and $\oplus,\ominus$, $\Mult$ and $\otimes$ for a better visual representation.
\vspace{-0.5\topsep}
\begin{align*}
	[[\vec{t}_k - \boldsymbol \rho_k]] &= \Mult \left ( \E\left(\frac{-1}{L}\vec{F}^\intercal \right),\E(\vec{x}(t)) \right) \oplus \Mult \left(\Add\left(\vec{I}_{Mm},\E\left( \frac{1}{L}\vec{H}\right)\right),\E\left(\vec{U}_k \right) \right) \oplus \\
	&\oplus \Mult \left(\Add\left(\E(\eta)\otimes \vec{I}_{Mm}, \E\left(\frac{-\eta}{L}\vec{H} \right)\right),\big(\E(\vec{U}_k) \ominus \E(\vec{U}_{k-1})\big) \right), \numberthis\label{eq:encrypted_iter}
\end{align*}
where $\boldsymbol\rho_k$ is the secret obtained by applying $f_\mr{iter}$ on the \LabHE secrets of the inputs of $f_\mr{iter}$. When the actuator applies the \LabHE decryption primitive on $[[\vec{t}_k - \boldsymbol \rho_k]]$, $\boldsymbol\rho_k$ is removed. Hence, for simplicity, we will write $[[\vec{t}_k]]$ instead of $[[\vec{t}_k - \boldsymbol \rho_k]]$.

Second, let us address how to perform~\eqref{eq:iter_U} in a private way. We have to perform the projection of $\vec{t}_k$ over the feasible domain described by $\vec{h}_u$ and $\vec{l}_u$, where all the variables are encrypted, as well as the private update of $\vec{U}_{k+1}$ with the projected iterate. One solution to the private comparison was described in Section~\ref{subsec:comparison}. The cloud has $[[\vec{t}_k]]$ and assume it also has the \AHE encryptions of the limits $[[\vec{h}_u]]$ and $[[\vec{l}_u]]$. 
The actuator has the master key of the \LabHE, as well as $\boldsymbol\rho_k$, that it computed offline. 
The cloud and the actuator will engage in two instances of the $\mr{DGK}$ protocol and oblivious transfer: first, to compare $\vec{t}_k$ to $\vec{h}_u$ and obtain an intermediate value of $\vec{U}_{k+1}$, and second, to compare $\vec{U}_{k+1}$ to $\vec{l}_u$ and to update the iterate $\vec{U}_{k+1}$. 

Before calling the comparison protocol~\ref{prot:encr_DGKV}, described in Section~\ref{subsec:comparison}, between $[[\vec{h}_u]]$ and $[[\vec{t}_k]]$, as well as between $[[\vec{l}_u]]$ and $[[\vec{U}_{k+1}]]$, the cloud should randomize the order of the inputs, such that, after obtaining the comparison bit, the actuator does not learn whether the iterate was feasible or not. This is done in lines 8 and 11 in Protocol~\ref{prot:iter_MPC}. 
Upon the completion of the comparison, the cloud and actuator perform the oblivious transfer variant, described in Section~\ref{subsec:OT}, such that the cloud obtains the intermediate value of $[[\vec{U}_{k+1}]]$ and subsequently, update the \AHE encryption of the iterate $[[\vec{U}_{k+1}]]$. At the last iteration, the cloud and actuator perform the standard oblivious transfer for the first $m$ positions in the values such that the actuator obtains $\vec{u}(t)$. 
Finally, because the next iteration can proceed only if the cloud has access to the full \LabHE encryption $\E(\vec{U}_{k+1})$, instead of $[[\vec{U}_{k+1}]]$, the cloud and actuator have to refresh the encryption. Specifically, the cloud secret-shares $[[\vec{U}_{k+1}]]$ in $[[\vec{U}_{k+1} - \vec{r}_{k+1}]]$ and $\vec{r}_{k+1}$, and sends $[[\vec{U}_{k+1} - \vec{r}_{k+1}]]$ to the actuator. The actuator decrypts it, and, using a previously generated secret, sends back $\E(\vec{U}_{k+1} - \vec{r}_{k+1}) = \big(\vec{U}_{k+1} - \vec{r}_{k+1}, \big[\big[\boldsymbol{\mf{b}}^U_{k+1}\big]\big]\big)$. Then, the cloud recovers the \LabHE encryption as $\E(\vec{U}_{k+1} ) = \Add(\E(\vec{U}_{k+1} - \vec{r}_{k+1}), \vec{r}_{k+1})$. 
In what follows, we will outline the private protocols obtained by integrating the above observations.

\vspace{-2\topsep}
\subsection{Private protocol}\label{subsec:protocol}
Assume that $K, N$ are fixed and known by all parties. Subscript $S^i$ stands for the $i$-th subsystem, for $i\in[M]$, subscript $Set$ for the Setup, subscript $A$ for the actuator and subscript $C$ for the cloud. 

\vspace{-0.5\topsep}
\begin{protocol}[Initialization of encrypted MPC]\label{prot:prior_MPC}
\small
\begin{algorithmic}[1]
\Require{Actuator: $f_{\mr{MPC}}$; Subsystems: $\mc U^i$; Setup: $\vec{A}, \vec{B}, \vec{P}, \vec{Q}, \vec{R}$.}
\Ensure{Subsytems: $\mr{mpk}, \mr{upk}_{S^i},\mr{usk}_{S^i},\boldsymbol{\tau}_{S^i}, \boldsymbol{\mf b}_{S^i}, [[\boldsymbol{\mf b}_{S^i}]], \mf{R}_{S^i}$; Setup: $\vec{H}, \vec{F}, \eta, L,$ $\mr{mpk},$ $\mr{upk}_{Set},$ $\mr{usk}_{Set},\boldsymbol{\tau}_{Set}, \boldsymbol{\mf b}_{Set}, [[\boldsymbol{\mf b}_{Set}]], \mf{R}_{Set}$; Actuator: $\mr{usk},\mr{upk}, \boldsymbol{\tau}_{A}, \boldsymbol{\mf b}_{A}, [[\boldsymbol{\mf b}_{A}]], \mf{R}_A$; Cloud: $\mr{mpk},$ \break$\E\left(-\frac{1}{L}\vec{H}\right),\E\left(-\frac{\eta}{L}\vec{H}\right)$, $\E\left(\frac{1}{L}\vec{F}^\intercal\right)$, $\E(\eta)$, $[[\vec{h}_u]]$, $[[\vec{l}_u]]$, $\mf{R}_C$.}
\Statex \underline{Offline}:
\State Actuator: Generate $(\mr{mpk},\mr{msk})\leftarrow \Init(1^\sigma)$ and distribute $\mr{mpk}$ to the others. Also generate a key $\mr{usk}$ for itself.
\State Subsystems, Setup: Each get $(\mr{usk},\mr{upk})\leftarrow \KeyGen(\mr{mpk})$ and send $\mr{upk}$ to the actuator.
\State Subsystems, Setup, Actuator: Allocate labels to the inputs of function $f_{\mr{MPC}}$: $\tau_1,\ldots,\tau_v$, where $v$ is the total number of inputs, as follows: 
	\Statex Subsystem $i$: for $\vec{x}^i(t)$ of size $n^i$, where $i$ denotes a subsystem, generate the corresponding labels $\tau_{\vec{x}^i(t)} = [\begin{matrix} 0 & 1 & n^i \ldots & n^i-1 \end{matrix}]^\intercal$.
	\Statex Setup: for matrix $\vec{H}\in\mathbb R^{Mm\times Mm}$, set $l = 0$, generate \break $\tau_{\vec{H}} = \left[\begin{smallmatrix} l & l+1 & \ldots & l+Mm-1 \\ \vdots & & & \vdots \\ l+(Mm-1)Mm & l+(Mm-1)Mm+1 & \ldots & l+ M^2m^2 -1  \end{smallmatrix}\right]$ and update $l = M^2m^2$, then follow the same steps for $\vec{F}$, starting from $l$ and updating it.
	\Statex Actuator: follow the same steps as the subsystems and setup, and then generate similar labels for the iterates $\vec{U}_k$ starting from the last $l$, for $k=0,\ldots,K-1$.
\State Subsystems, Setup, Actuator, Cloud: Generate randomness for blinding and encryptions $\mf{R}$. 
\State Subsystems, Setup, Actuator: Perform the offline part of the $\mr{LabHE}$ encryption primitive. The actuator also performs the offline part for the decryption. The parties thus obtain $\boldsymbol{\mf b}, [[\boldsymbol{\mf b}]]$. 
\State Actuator: Generate initializations for the initial iterate $\vec{U}'_0$.
\State Actuator: Form the program $\mc P=(f_{\mr{MPC}},\tau_1,\ldots,\tau_v)$. 
\Statex \underline{Initialization}:
\State Setup: Compute $\vec{H}$ and $\vec{F}$ from $\vec{A}, \vec{B}, \vec{P}, \vec{Q}, \vec{R}$ and then $L = \lambda_{\mr{max}}(\vec{H})$ and $\eta =$ \break$ (\sqrt{\kappa(\vec{H})}-1)/(\sqrt{\kappa(\vec{H})}+1)$. Perform the online part of $\mr{LabHE}$ encryption and send to the cloud: $\E\left(-\frac{1}{L}\vec{H}\right)$, $\E\left(-\frac{\eta}{L}\vec{H}\right)$, $\E\left(\frac{1}{L}\vec{F}^\intercal\right)$, $\E(\eta)$.
\State Subsystems: Perform the online part of $\mr{LabHE}$ encryption and send to the cloud, which aggregates what it receives into: $[[\vec{h}_u]], [[\vec{l}_u]]$.
\vspace{-1pt}
\end{algorithmic}
\end{protocol}

\begin{protocol}[Encrypted MPC step]\label{prot:iter_MPC}
\small
\begin{algorithmic}[1]
\Require{Actuator: $f_{\mr{MPC}}$; Subsystems: $\vec{x}^i(t)$, $\mc U^i$; Setup: $\vec{A}, \vec{B}, \vec{P}, \vec{Q}, \vec{R}$.}
\Ensure{Actuator: $\vec{u}(t)$}
\Statex \underline{Offline + Initialization}:
\State Subsystems, Setup, Cloud, Actuator: Run Protocol~\ref{prot:prior_MPC}.
\Statex \underline{Online}:
	\State Cloud: $\left[\left[\frac{1}{L}\vec{F}^\intercal \vec{x}(t)\right]\right] \leftarrow\Mult \left(\E\left(\frac{1}{L}\vec{F}^\intercal\right), \E(\vec{x}(t))\right)$.
		\State Actuator: Send the initial iterate to the cloud: $\E(\vec{U}'_0)$.
		\State Cloud: Change the initial iterate: $\E(\vec{U}_0) = \Add\left(\E\left(\vec{U}'_0\right), \vec{r}_0\right)$. 
	\State Cloud: $\E \left(\vec{U}_{-1}\right)\leftarrow \E\left(\vec{U}_0\right)$.
	\For{$k=0,\ldots,K-1$}
		\State Cloud: Compute $[[\vec{t}_k]]$ as in Equation~\eqref{eq:encrypted_iter}.
		\State Cloud: $\left([[\vec{a}_k]], [[\vec{b}_k]]\right) \leftarrow \mr{randomize}\left([[\vec{h}_u]],[[\vec{t}_k]]\right)$.
		\State Cloud, Actuator: Execute comparison protocol~\ref{prot:encr_DGKV}; Actuator obtains $\boldsymbol \delta_k$.
		\State Cloud, Actuator: $\left[\left[\vec{U}_{k+1}\right]\right]\leftarrow \mr{OT'}\left([[\vec{a}_k]], [[\vec{b}_k]],\boldsymbol \delta_k,\mr{msk} \right)$.
		\State Cloud: $\left([[\vec{a}_k]], [[\vec{b}_k]]\right) \leftarrow \mr{randomize}\left([[\vec{l}_u]],\left[\left[ \vec{U}_{k+1} \right]\right]\right)$.
		\State Cloud, Actuator: Execute comparison protocol~\ref{prot:encr_DGKV}; Actuator obtains $\boldsymbol \delta_k$.
		\If{k!=K-1}
			\State Cloud, Actuator: $\left[\left[\vec{U}_{k+1}\right]\right]\leftarrow \mr{OT'}\left([[\vec{a}_k]], [[\vec{b}_k]],\boldsymbol \delta_k \veebar \vec{1},\mr{msk}\right)$. Cloud receives $\left[\left[\vec{U}_{k+1}\right]\right]$.
			\State Cloud: Send to the actuator $\left[\left[\vec{U}'_{k+1}\right]\right]\leftarrow\Add\left(\left[\left[ \vec{U}_{k+1} \right]\right],\left[\left[-\vec{r}_k\right]\right]\right)$, where $\vec{r}_k$ is 
			\Statex\hskip\algorithmicindent\hskip\algorithmicindent randomly selected from $\mc M^{Mm}$.
			\State Actuator: Decrypt $\left[\left[\vec{U}'_{k+1}\right]\right]$ and send back $\E(\vec{U}'_{k+1})$.
			\State Cloud: $\E\left(\vec{U}_{k+1}\right) \leftarrow \Add\left(\E\left(\vec{U}'_{k+1}\right), \vec{r}_k\right)$.
		\Else
			\State Cloud, Actuator: $\vec{u}(t)\leftarrow \mr{OT}\left([[\vec{a}_k]]_{1:m}, [[\vec{b}_k]]_{1:m},\{\boldsymbol \delta_k\}_{1:m} \veebar \vec{1},\mr{msk}\right)$. Actuator 
			\Statex\hskip\algorithmicindent\hskip\algorithmicindent receives $\vec{u}(t)$.
		\EndIf
	\EndFor
	\State Actuator: Input $\vec{u}(t)$ to the system.
\end{algorithmic}
\end{protocol}
Lines 3 and 4 ensure that neither the cloud nor the actuator knows the initial point of the optimization problem.

\subsection{Privacy of Protocol~\ref{prot:iter_MPC}}\label{sec:proof}
\begin{assumption}\label{assum:coalitions}\Assumption
An adversary cannot corrupt at the same time both the cloud controller and the virtual actuator or more than $M-1$ subsystems.

\end{assumption}

\begin{theorem}\label{thm:privacy}
Under Assumption~\ref{assum:coalitions}, the encrypted MPC solution presented in Protocol~\ref{prot:iter_MPC} achieves multi-party privacy (Definition~\ref{def:coalition_view}).
\end{theorem}

\begin{proof}
The components of the protocol: $\mr{AHE}$, secret sharing, pseudorandom generator, $\mr{LabHE}$, oblivious transfer and the comparison protocol are individually secure, meaning either their output is computationally indistinguishable from a random output or they already satisfy Definition~\ref{def:coalition_view}. We are going to build the views of the allowed coalitions and their corresponding simulators and use the previous results to prove they are computationally indistinguishable.

The cloud has no inputs and no outputs, hence its view is composed solely from received messages and coins:
\begin{align*}
	V_{C}(\emptyset) = \Bigg(&\mr{mpk}, \E\left(-\frac{1}{L}\vec{H}\right),\E\left(-\frac{\eta}{L}\vec{H}\right),\E\left(\frac{1}{L}\vec{F}^\intercal\right),\E(\eta),[[\vec{h}_u]],[[\vec{l}_u]], \mf{R}_C,\E(\vec{U}'_0), \\
	&\Big\{\E(\vec{U}_k), [[\vec{a}_k]], [[\vec{b}_k]], \left[\left[U_{k+1}\right]\right], \mr{msg}_{\mr{Pr}.\ref{prot:encr_DGKV}}, \mr{msg}_{\mr{OT}} \Big\}_{k\in\{0,\ldots,K-1\}} \Bigg). \numberthis\label{eq:view_cloud}
\end{align*}

The actuator's input is the function $f_{\mr{MPC}}$ and the output is $\vec{u}(t)$. Then, its view is:
\begin{align}\label{eq:view_actuator}
\begin{split}
	V_{A}(f_{\mr{MPC}}) = \Bigg(f_{\mr{MPC}}, \mr{mpk}, \mr{msk}, \mr{upk}, \mf{R}_A,
	\Big\{\vec{U}'_{k+1}, \mr{msg}_{\mr{Pr}.\ref{prot:encr_DGKV}}, \mr{msg}_{\mr{OT}} \Big\}_{k\in\{0,\ldots,K-1\}},\vec{u}(t)  \Bigg),
\end{split}
\end{align}
which includes the keys $\mr{mpk}, \mr{msk}$ because their generation involve randomness.

The setup's inputs are the model and costs of the system and no output after the execution of Protocol~\ref{prot:iter_MPC}, since it is just a helper entity. Its view is:
\begin{align}\label{eq:view_setup}
\begin{split}
	V_{Set}(\vec{A}, \vec{B}, \vec{P}, \vec{Q}, \vec{R}) = \Bigg(\vec{A}, \vec{B}, \vec{P}, \vec{Q}, \vec{R}, \mr{mpk},\mr{usk}_{Set},\mf{R}_{Set} \Bigg).
\end{split}
\end{align}

Finally, for a subsystem~$i$, $i\in[M]$, the inputs are the local control action constraints and the measured states and there is no output obtained through computation after the execution of Protocol~\ref{prot:iter_MPC}. Its view is:
\begin{align}\label{eq:view_subsystem}
\begin{split}
	V_{S^i}(\mc U^i, \vec{x}^i(t)) = \Bigg(\mc U^i, \vec{x}^i(t), \mr{mpk},\mr{usk}_{S^i},\mf{R}_{S^i}\Bigg).
\end{split}
\end{align}

In general, the indistinguishability between the view of the adversary corrupting the real-world parties and the simulator is proved through sequential games in which some real components of the view are replaced by components that could be generated by the simulator, which are indistinguishable from each other. In our case, we can directly make the leap between the real view and the simulator by showing that the cloud only receives encrypted messages, and the actuator receives only messages blinded by one-time pads. In~\cite{Alexandru2018cloudQPHE}, the proof for the privacy of a quadratic optimization problem solved in the same architecture is given with sequential games.

For the cloud, consider a simulator $S_C$ that generates $\widetilde{\mr{mpk}},\widetilde{\mr{msk}}\leftarrow \Init(1^\sigma)$, generates $\widetilde{\mr{usk}}_j\leftarrow\KeyGen(\widetilde{\mr{mpk}})$, for $j\in\{S^i, Set, A\}$, $i\in[M]$ and then $(\widetilde{\boldsymbol\tau}, \widetilde{\boldsymbol{\mf{b}}}, \widetilde{[[\boldsymbol{\mf{b}}]]})_j$. We use $\widetilde{(\cdot)}$ also over the encryptions to show that the keys are different from the ones in the view. Subsequently, the simulator encrypts random values of appropriate sizes to obtain $\widetilde{\E\left(-\frac{1}{L}\vec{H}\right)},\widetilde{\E\left(-\frac{\eta}{L}\vec{H}\right)},\widetilde{\E\left(\frac{1}{L}\vec{F}^\intercal\right)},\widetilde{\E(\eta)},\widetilde{[[\vec{h}_u]]},\widetilde{[[\vec{l}_u]]}, \widetilde{\E(\vec{U}'_0})$. $S_C$ generates the coins $\widetilde{\mf R}_C$ as in line~4 in Protocol~\ref{prot:prior_MPC} and obtains $\E(\vec{U}_0)$ as in line~4 in Protocol~\ref{prot:iter_MPC}. Then, for each $k=\{0,\ldots,K-1\}$, it computes $\widetilde{[[\vec{t}_k]]}$ as in line~7 in Protocol~\ref{prot:iter_MPC} and shuffles $\widetilde{[[\vec{t}_k]]}$ and $\widetilde{[[\vec{h}_u]]}$ into $\widetilde{[[\vec{a}_k]]}, \widetilde{[[\vec{b}_k]]}$. $S_C$ then performs the same steps as the simulator for party A in Protocol~\ref{prot:encr_DGKV} and gets $\widetilde{\mr{msg}}_{\mr{Pr}.\ref{prot:encr_DGKV}}$. Furthermore, $S_C$ generates an encryption of random bits $\widetilde{\boldsymbol \delta_k}$ and of $\widetilde{\E(\vec{U}_k)}$ and performs the same steps as the simulator for party A as in the proof of Proposition~\ref{prop:OT} (or the simulator for the standard OT) and gets $\widetilde{\mr{msg}}_{\mr{OT}}$. It then outputs:
\begin{align*}
	S_{C}(\emptyset) =\Bigg(&\widetilde{\mr{mpk}}, \widetilde{\E\left(-\frac{1}{L}\vec{H}\right)},\widetilde{\E\left(-\frac{\eta}{L}\vec{H}\right)},\widetilde{\E\left(\frac{1}{L}\vec{F}^\intercal\right)},\widetilde{\E(\eta)},\widetilde{[[\vec{h}_u]]},\widetilde{[[\vec{l}_u]]}, \widetilde{\mf{R}}_C, \widetilde{\E(\vec{U}'_0}),\\
	&\Big\{\widetilde{\E(\vec{U}_k)}, \left(\widetilde{[[\vec{a}_k]]}, \widetilde{[[\vec{b}_k]]}\right), \widetilde{\left[\left[U_{k+1}\right]\right]}, \widetilde{\mr{msg}}_{\mr{Pr}.\ref{prot:encr_DGKV}}, \widetilde{\mr{msg}}_{\mr{OT}} \Big\}_{k\in\{0,\ldots,K-1\}} \Bigg). \numberthis\label{eq:sim_cloud}
\end{align*}

All the values in the view of the cloud in~\eqref{eq:view_cloud} -- with the exception of the random values $\mf{R}_C$ and the key $\mr{mpk}$, which are statistically indistinguishable from $\widetilde{\mf{R}}_C$ and $\widetilde{\mr{mpk}}$ because they are drawn from the same distributions -- are encrypted with semantically secure encryptions schemes (\AHE and $\mr{LabHE}$). This means they are computationally indistinguishable from the encryptions of random values in~\eqref{eq:sim_cloud}, even with different keys. This happens even when the values from different iterations are encryptions of correlated quantities. 
Thus, $V_C(\emptyset) \stackrel{c}{\equiv} S_C(\emptyset)$.

We now build a simulator $S_A$ for the actuator that takes as input $f_{\mr{MPC}},\vec{u}(t)$. $S_A$ will take the same steps as in lines~1, 3--7 in Protocol~\ref{prot:prior_MPC}, obtaining $\widetilde{\mr{mpk}}, \widetilde{\mr{msk}}, \widetilde{\mr{upk}}, \widetilde{\mr{usk}}, \widetilde{\boldsymbol{\tau}}_{A},$ $\widetilde{\boldsymbol{\mf b}}_{A}, \widetilde{[[\boldsymbol{\mf b}_{A}]]},\widetilde{\mf{R}}_A,$ $\widetilde{\vec{U}'_0}$ and instead of line 2, it generates $\widetilde{\mr{upk}}_j, \widetilde{\mr{usk}}_j\leftarrow\KeyGen(\widetilde{\mr{mpk}})$ itself, for $j\in\{S^i, Set\}$, $i\in[M]$. For $k=0,\ldots,K-1$, $S_A$ performs the same steps as the simulator for party B in Protocol~\ref{prot:encr_DGKV} and gets $\widetilde{\mr{msg}}_{\mr{Pr}.\ref{prot:encr_DGKV}}$. Furthermore, $S_A$ performs the same steps as the simulator for party B as in the proof of Proposition~\ref{prop:OT} (or the simulator for the standard OT) and gets $\widetilde{\mr{msg}}_{\mr{OT}}$. It then outputs:

\begin{align}\label{eq:sim_actuator}
\begin{split}
	S_{A}(f_{\mr{MPC}},\vec{u}(t)) = \Bigg( &f_{\mr{MPC}}, \widetilde{\mr{mpk}}, \widetilde{\mr{msk}}, \widetilde{\mr{upk}}, \widetilde{\mf{R}}_A, \\
	&\Big\{\widetilde{\vec{U}'_{k+1}}, \widetilde{\mr{msg}}_{\mr{Pr}.\ref{prot:encr_DGKV}}, \widetilde{\mr{msg}}_{\mr{OT}} \Big\}_{k\in\{0,\ldots,K-1\}}, \vec{u}(t) \Bigg).
\end{split}
\end{align}

All the values in the view of the actuator in~\eqref{eq:view_actuator} -- with the exception of the random values $\mf{R}_A$ and the keys $\mr{upk}$, which are statistically indistinguishable from $\widetilde{\mf{R}}_A$ and $\widetilde{\mr{upk}}$ because they are drawn from the same distributions and $\vec{u}(t)$ -- are blinded by random numbers, different at every iteration, which means that they are statistically indistinguishable from the random values in~\eqref{eq:sim_actuator}. This again holds even when the values that are blinded at different iterations are correlated and the actuator knows the solution $\vec{u}(t)$, because the values of interest are drowned in large noise. Thus, $V_{A}(f_{\mr{MPC}})\stackrel{c}{\equiv} S_{A}(f_{\mr{MPC}},\vec{u}(t))$.

The setup and subsystems do not receive any other messages apart from the master public key~\eqref{eq:view_setup},~\eqref{eq:view_subsystem}. Hence, a simulator $S_{Set}$ for the setup and a simulator $S_{S^i}$ for a subsystem~$i$ can simply generate $\widetilde{\mr{mpk}}\leftarrow \Init(1^\sigma)$ and then proceed with the execution of lines 2--5 in Protocol~\ref{prot:prior_MPC} and output their inputs, messages and coins. The outputs of the simulators are trivially indistinguishable from the views.

When an adversary corrupts a coalition, the view of the coalition contains the inputs of all parties, and a simulator takes the coalition's inputs and outputs. The view of the coalition between the cloud, the setup, and a number $l$ of subsystems is:
\begin{align*}
V_{CSl}\big(\vec{A}, \vec{B}, \vec{P}, \vec{Q}, \vec{R}, &\{\mc U^i, \vec{x}^i(t)\}_{i\in{i_1,\ldots,i_l}}\big)  = V_C(\emptyset)\cup V_{Set}(\vec{A}, \vec{B}, \vec{P}, \vec{Q}, \vec{R}) \cup \\
\cup V_{S^{i_1}} & (\mc U^{i_1}, \vec{x}^{i_1}(t)) \cup \ldots \cup V_{S^{i_l}}(\mc U^{i_l}, \vec{x}^{i_l}(t)).
\end{align*}

A simulator $S_{CSl}$ for this coalition takes in the inputs of the coalition and no output and performs almost the same steps as $S_C$, $S_{Set}$, $S_{S^i}$, without randomly generating the quantities that are known by the coalition. The same argument of having the messages drawn from the same distributions and encrypted with semantically secure encryption schemes proves the indistinguishability between $V_{CSl}(\cdot)$ and $S_{CSl}(\cdot)$.

The view of the coalition between the actuator, the setup, and a number $l$ of subsystems is the following:
\begin{align*}
V_{ASl}\big(f_{\mr{MPC}},\vec{u}(t),~&\vec{A}, \vec{B}, \vec{P}, \vec{Q}, \vec{R}, \{\mc U^i, \vec{x}^i(t)\}_{i\in{i_1,\ldots,i_l}}\big)  = V_A(f_{\mr{MPC}},\vec{u}(t))\cup \\
\cup V_{Set}( & \vec{A}, \vec{B}, \vec{P}, \vec{Q}, \vec{R})
\cup V_{S^{i_1}}  (\mc U^{i_1}, \vec{x}^{i_1}(t)) \cup \ldots \cup V_{S^{i_l}}(\mc U^{i_l}, \vec{x}^{i_l}(t)).
\end{align*}

A simulator $S_{ASl}$ for this coalition takes in the inputs of the coalition and $\vec{u}(t)$ and performs almost the same steps as $S_A$, $S_{Set}$, $S_{S^i}$, without randomly generating the quantities that are now known. The same argument of having the messages drawn from the same distributions and blinded with one-time pads proves the indistinguishability between $V_{ASl}(\cdot)$ and $S_{ASl}(\cdot)$. 

The proof is now complete.
\end{proof}

We can also have the private MPC scheme run for multiple time steps. Protocol~\ref{prot:prior_MPC} can be modified to also generate the labels and secrets for $T$ time steps. Protocol~\ref{prot:iter_MPC} can be run for multiple time steps, and warm starts can be included by adding two lines such that the cloud obtains $\E\left(\left\{\vec{U}^t_K\right\}_{m+1:M}\right)$ and sets $\E(\vec{U}^{t+1}_0) = \left[ \begin{matrix}{\vec{U}^{t}_K}^\intercal & \vec{0}_m^\intercal \end{matrix}\right]^\intercal$.

\subsection{Analysis of errors}\label{subsec:error_analysis}
As mentioned at the beginning of the section, the values that are encrypted, added to or multiplied with encrypted values have to be integers. We consider \mbox{fixed-point} representations with one sign bit, $l_i$ integer bits and $l_f$ fractional bits. We multiply the values by $2^{l_f}$ then truncate to obtain integers prior to encryption, and, after decryption, we divide the result by the appropriate quantity (e.g., we divide the result of a multiplication by $2^{2l_f}$). Furthermore, the operations can increase the number of bits of the result, hence, before the comparisons in Protocol~\ref{prot:iter_MPC} an extra step that performs interactive truncation has to be performed, because Protocol~\ref{prot:encr_DGKV} requires a fixed number of bits for the inputs. Also, notice that when the encryption is refreshed, i.e., a ciphertext is decrypted and re-encrypted, the accumulation of bits due to the operations is truncated back to the desired size.

Working with \mbox{fixed-point} representations can lead to overflow, quantization and arithmetic \mbox{round-off} errors. Thus, we want to compute the deviation between the fixed-point solution and optimal solution of the FGM algorithm in~\eqref{eq:FGM}. In order to bound it, we need to ensure that the number of fractional bits $l_f$ is sufficiently large such that the feasibility of the \mbox{fixed-point} precision solution is preserved, that the strong convexity of the \mbox{fixed-point} objective function still holds and that the \mbox{fixed-point} step size is such that the FGM converges. The errors can be written as states of a stable linear system with bounded disturbances. Bounds on the errors for the case of public model are derived in~\cite{Alexandru2018cloud} and similar bounds can be obtained for the private model. The bounds on this deviation can be used in an offline step to choose an appropriate \mbox{fixed-point} precision for the desired performance of the system. 

\section{Discussion}
\label{sec:discussion}
Secure multi-party computation protocols require many rounds of communication in general, and the amount of communication depends on the amount of data that needs to be concealed. 
This can also be observed in Protocol~\ref{prot:iter_MPC}. Therefore, in order to be able to use this proposed protocol, we need fast and reliable communication between the cloud and the actuator. 

In the architecture we considered, the subsystems are computationally and memory constrained devices, hence, they are only required to generate the encryptions for their measurements. The setup only holds constant data, and it only has to compute the matrices in~\eqref{eq:FGM} and encrypt them in the initialization step. Furthermore, we considered the existence of the setup entity for clarity in the exposition of the protocols, but the data held by the setup could be distributed to the other participants in the computation. In this case, the cloud would have to perform some extra steps in order to aggregate the encrypted system parameters (see~\cite{Alexandru2019encrypted} for a related solution). Notice that the subsystems and setup do not need to generate labels for the number of iterations, only the actuator does. The actuator is supposed to be a machine with enough memory to store the labels and reasonable amount of computation power such that the encryptions and decryptions are performed in a practical amount of time (that will be dependent on the sampling time of the system), but less powerful than the cloud. The cloud controller is assumed to be a server, which has enough computational power and memory to be capable to deal with the computations on the ciphertexts, which can be large, depending on the encryption schemes employed. 

If fast and reliable communication is not available or if the actuator is a highly constrained device, then a fully homomorphic encryption solution that is solely executed by the cloud might be more appropriate, although its execution can be substantially slower.

Compared to the two-server MPC with private signals but public model from~\cite{Alexandru2018cloud}, where only \AHE is required, the MPC with private signals and private model we considered in this chapter is only negligibly more expensive. Specifically, the ciphertexts are augmented with one secret that has the number of bits substantially smaller than the number of bits in an \AHE ciphertext, and each online iteration only incurs one extra round of communication, one decryption and one encryption. All the other computations regarding the secrets are done offline. This shows the efficiency and suitability of \LabHE for encrypted control applications.

\bibliographystyle{spmpsci.bst}
\bibliography{biblo}

\end{document}